\documentclass[acmsmall]{acmart}

\makeatletter
\@ifundefined{lhs2tex.lhs2tex.sty.read}%
  {\@namedef{lhs2tex.lhs2tex.sty.read}{}%
   \newcommand\SkipToFmtEnd{}%
   \newcommand\EndFmtInput{}%
   \long\def\SkipToFmtEnd#1\EndFmtInput{}%
  }\SkipToFmtEnd

\newcommand\ReadOnlyOnce[1]{\@ifundefined{#1}{\@namedef{#1}{}}\SkipToFmtEnd}
\usepackage{amstext}
\usepackage{amssymb}
\usepackage{stmaryrd}
\DeclareFontFamily{OT1}{cmtex}{}
\DeclareFontShape{OT1}{cmtex}{m}{n}
  {<5><6><7><8>cmtex8
   <9>cmtex9
   <10><10.95><12><14.4><17.28><20.74><24.88>cmtex10}{}
\DeclareFontShape{OT1}{cmtex}{m}{it}
  {<-> ssub * cmtt/m/it}{}

\DeclareFontShape{OT1}{cmtt}{bx}{n}
  {<5><6><7><8>cmtt8
   <9>cmbtt9
   <10><10.95><12><14.4><17.28><20.74><24.88>cmbtt10}{}
\DeclareFontShape{OT1}{cmtex}{bx}{n}
  {<-> ssub * cmtt/bx/n}{}

\newcommand{\Conid}[1]{\mathit{#1}}
\newcommand{\Varid}[1]{\mathit{#1}}
\newcommand{\anonymous}{\kern0.06em \vbox{\hrule\@width.5em}}

\usepackage{polytable}

\@ifundefined{mathindent}%
  {\newdimen\mathindent\mathindent\leftmargini}%
  {}%

\def\resethooks{%
  \global\let\SaveRestoreHook\empty
  \global\let\ColumnHook\empty}
\newcommand*{\savecolumns}[1][default]%
  {\g@addto@macro\SaveRestoreHook{\savecolumns[#1]}}
\newcommand*{\restorecolumns}[1][default]%
  {\g@addto@macro\SaveRestoreHook{\restorecolumns[#1]}}
\newcommand*{\aligncolumn}[2]%
  {\g@addto@macro\ColumnHook{\column{#1}{#2}}}

\resethooks

\newcommand{\onelinecommentchars}{\quad-{}- }
\newcommand{\commentbeginchars}{\enskip\{-}
\newcommand{\commentendchars}{-\}\enskip}

\newcommand{\visiblecomments}{%
  \let\onelinecomment=\onelinecommentchars
  \let\commentbegin=\commentbeginchars
  \let\commentend=\commentendchars}

\newcommand{\invisiblecomments}{%
  \let\onelinecomment=\empty
  \let\commentbegin=\empty
  \let\commentend=\empty}

\visiblecomments

\newlength{\blanklineskip}
\newcommand{\hsindent}[1]{\quad}
\let\hspre\empty
\let\hspost\empty

\EndFmtInput
\makeatother
\ReadOnlyOnce{forall.fmt}%
\makeatletter

\let\HaskellResetHook\empty
\newcommand*{\AtHaskellReset}[1]{%
  \g@addto@macro\HaskellResetHook{#1}}
\newcommand*{\HaskellReset}{\HaskellResetHook}

\newcommand\hsforall{\global\let\hsdot=\hsperiodonce}
\newcommand*\hsperiodonce[2]{#2\global\let\hsdot=\hscompose}
\newcommand*\hscompose[2]{#1}

\AtHaskellReset{\global\let\hsdot=\hscompose}

\HaskellReset

\makeatother
\EndFmtInput
\ReadOnlyOnce{polycode.fmt}%
\makeatletter

\newcommand{\hsnewpar}[1]%
  {{\parskip=0pt\parindent=0pt\par\vskip #1\noindent}}

\newcommand{\hscodestyle}{}

\newcommand{\sethscode}[1]%
  {\expandafter\let\expandafter\hscode\csname #1\endcsname
   \expandafter\let\expandafter\endhscode\csname end#1\endcsname}

  {\par\noindent
   \advance\leftskip\mathindent
   \hscodestyle
   \let\\=\@normalcr
   \let\hspre\(\let\hspost\)%
   \pboxed}%
  {\endpboxed\)%
   \par\noindent
   \ignorespacesafterend}

  {\hsnewpar\abovedisplayskip
   \advance\leftskip\mathindent
   \hscodestyle
   \let\hspre\(\let\hspost\)%
   \pboxed}%
  {\endpboxed%
   \hsnewpar\belowdisplayskip
   \ignorespacesafterend}

  {\hsnewpar\abovedisplayskip
   \advance\leftskip\mathindent
   \hscodestyle
   \let\\=\@normalcr
   \(\pboxed}%
  {\endpboxed\)%
   \hsnewpar\belowdisplayskip
   \ignorespacesafterend}

\newcommand{\plainhs}{\sethscode{plainhscode}}

\plainhs

  {\hsnewpar\abovedisplayskip
   \advance\leftskip\mathindent
   \hscodestyle
   \let\\=\@normalcr
   \(\parray}%
  {\endparray\)%
   \hsnewpar\belowdisplayskip
   \ignorespacesafterend}

  {\parray}{\endparray}

  {\(\parray}{\endparray\)}

\def\codeframewidth{\arrayrulewidth}
\RequirePackage{calc}

  {\parskip=\abovedisplayskip\par\noindent
   \hscodestyle
   \arrayrulewidth=\codeframewidth
   \tabular{@{}|p{\linewidth-2\arraycolsep-2\arrayrulewidth-2pt}|@{}}%
   \hline\framedhslinecorrect\\{-1.5ex}%
   \let\endoflinesave=\\
   \let\\=\@normalcr
   \(\pboxed}%
  {\endpboxed\)%
   \framedhslinecorrect\endoflinesave{.5ex}\hline
   \endtabular
   \parskip=\belowdisplayskip\par\noindent
   \ignorespacesafterend}

\newcommand{\framedhslinecorrect}[2]%
  {#1[#2]}

  {\(\def\column##1##2{}%
   \let\>\undefined\let\<\undefined\let\\\undefined
   \newcommand\>[1][]{}\newcommand\<[1][]{}\newcommand\\[1][]{}%
   \def\fromto##1##2##3{##3}%
   }{\) }%

  {\let\orighscode=\hscode
   \let\origendhscode=\endhscode
   \def\endhscode{\def\hscode{\endgroup\def\@currenvir{hscode}\\}\begingroup}
   \orighscode\def\hscode{\endgroup\def\@currenvir{hscode}}}%
  {\origendhscode
   \global\let\hscode=\orighscode
   \global\let\endhscode=\origendhscode}%

\makeatother
\EndFmtInput

\ReadOnlyOnce{Formatting.fmt}%
\makeatletter

\let\Varid\mathit
\let\Conid\mathsf

\def\commentbegin{\quad\begingroup\color{teal}\{\ }
\def\commentend{\}\endgroup}

\makeatother
\EndFmtInput
\usepackage{amsmath}
\usepackage{mathpartir}
\usepackage{amsfonts}
\usepackage{stmaryrd}
\usepackage{hyperref}
\usepackage{scalerel}
\usepackage{url}
\usepackage{subfig}
\usepackage{enumitem}
\usepackage{multicol}
\usepackage{graphicx}

\usepackage{doubleequals}

\newcommand{\delete}[1]{}

\makeatother

\setcopyright{none}             

\allowdisplaybreaks

\definecolor{airforceblue}{rgb}{0.36, 0.54, 0.66}
\definecolor{bleudefrance}{rgb}{0.19, 0.55, 0.91}
\definecolor{blue(ncs)}{rgb}{0.0, 0.53, 0.74}
\definecolor{mediumpersianblue}{rgb}{0.0, 0.4, 0.65}
\everymath{\color{mediumpersianblue}}

\begin{document}

\title[Equational Reasoning for Non-determinism Monad]%
{Equational Reasoning for Non-determinism Monad:}
\subtitle{The Case of Spark Aggregation}

\author{Shin-Cheng Mu}
\affiliation{
  \department{Institute of Information Science}              
  \institution{Academia Sinica}            
  \country{Taiwan}
}
\email{scm@iis.sinica.edu.tw}          

\begin{abstract}
As part of the author's studies on equational reasoning for monadic programs, this report focus on non-determinism monad.
We discuss what properties this monad should satisfy, what additional operators and notations can be introduced to facilitate equational reasoning about non-determinism, and put them to the test by proving a number of properties in our example problem inspired by the author's previous work on proving properties of Spark aggregation.
\end{abstract}

\maketitle

\section{Introduction}
\label{sec:intro}

In functional programming, \emph{pure} programs are those that can be understood as static mappings from inputs to outputs.
The main advantage of staying in the pure realm is that properties of pure entities can be proved by equational reasoning.
Side effects, in contrast, used to be considered the ``awkward squad'' that are difficult to be reasoned about.
\citet{GibbonsHinze:11:Just}, however, showed that effectful, monadic programs may also be reasoned about in a mathematical manner, using monad laws and properties of effect operators.

This report is part of a series of the author's studies on equational reasoning for monadic programs.
In this report we focus on non-determinism monad --- in our definition that is a monad having two effect operators, one allowing a program to fail, another allowing a non-deterministic choice between two results.
We discuss what properties these operators should satisfy, what additional operators and notations can be introduced to facilitate equational reasoning of this monad, and put them to the test by proving a number of properties in our example problem: Spark aggregation.

Much of this report is inspired by the author's joint work with \citet{Lengal:17:Executable}, in which we formalised Spark, a platform for distributed computation, and derived properties under which a distributed Spark aggregation represents a \emph{deterministic} computation.
Therefore, many examples in this report are about finding out when processing a non-deterministic permutation (simulating arbitrary distribution of data) produces a deterministic result.
\section{Monad and Non-determinism}

A monad consists of a type constructor \ensuremath{\Conid{M}\mathbin{::}\mathbin{*}\to \mathbin{*}} and two operators \ensuremath{\Varid{return}\mathbin{::}\Varid{a}\to \Conid{M}\;\Varid{a}} and ``bind'' \ensuremath{(\mathbin{\hstretch{0.7}{=\!\!<\!\!<}})\mathbin{::}(\Varid{a}\to \Conid{M}\;\Varid{b})\to \Conid{M}\;\Varid{a}\to \Conid{M}\;\Varid{b}} that satisfy the following {\em monad laws}:
\begin{align}
  \ensuremath{\Varid{f}\mathbin{\hstretch{0.7}{=\!\!<\!\!<}}\Varid{return}\;\Varid{x}} &= \ensuremath{\Varid{f}\;\Varid{x}}\mbox{~~,} \label{eq:monad-bind-ret}\\
  \ensuremath{\Varid{return}\mathbin{\hstretch{0.7}{=\!\!<\!\!<}}\Varid{m}} &= \ensuremath{\Varid{m}} \mbox{~~,} \label{eq:monad-ret-bind}\\
  \ensuremath{\Varid{f}\mathbin{\hstretch{0.7}{=\!\!<\!\!<}}(\Varid{g}\mathbin{\hstretch{0.7}{=\!\!<\!\!<}}\Varid{m})} &= \ensuremath{(\lambda \Varid{x}\to \Varid{f}\mathbin{\hstretch{0.7}{=\!\!<\!\!<}}\Varid{g}\;\Varid{x})\mathbin{\hstretch{0.7}{=\!\!<\!\!<}}\Varid{m}} \mbox{~~.}
    \label{eq:monad-assoc}
\end{align}
Rather than the usual \ensuremath{(\mathrel{\hstretch{0.7}{>\!\!>\!\!=}})\mathbin{::}\Conid{M}\;\Varid{a}\to (\Varid{a}\to \Conid{M}\;\Varid{b})\to \Conid{M}\;\Varid{b}}, in the laws above we use the reversed bind \ensuremath{(\mathbin{\hstretch{0.7}{=\!\!<\!\!<}})}, which is consistent with the direction of function composition and more readable when we program in a style that uses composition.
When we use bind with $\lambda$-abstractions, it is more natural to write \ensuremath{\Varid{m}\mathrel{\hstretch{0.7}{>\!\!>\!\!=}}\lambda \Varid{x}\to \Varid{f}\;\Varid{x}}.
In this report we use the former more than the latter, thus the choice of notation.
We also define \ensuremath{\Varid{m}_{1}\mathbin{\hstretch{0.7}{<\!\!<}}\Varid{m}_{2}\mathrel{=}\Varid{const}\;\Varid{m}_{1}\mathbin{\hstretch{0.7}{=\!\!<\!\!<}}\Varid{m}_{2}}. Note that \ensuremath{(\mathbin{\hstretch{0.7}{>\!\!>}})} has type \ensuremath{\Conid{M}\;\Varid{a}\to \Conid{M}\;\Varid{b}\to \Conid{M}\;\Varid{b}}.

\begin{figure}
\begin{hscode}\SaveRestoreHook
\column{B}{@{}>{\hspre}l<{\hspost}@{}}%
\column{8}{@{}>{\hspre}l<{\hspost}@{}}%
\column{E}{@{}>{\hspre}l<{\hspost}@{}}%
\>[B]{}(\mathrel{\hstretch{0.7}{<\!\!\!=\!\!<}}){}\<[8]%
\>[8]{}\mathbin{::}(\Varid{b}\to \Conid{M}\;\Varid{c})\to (\Varid{a}\to \Conid{M}\;\Varid{b})\to \Varid{a}\to \Conid{M}\;\Varid{c}{}\<[E]%
\\
\>[B]{}(\Varid{f}\mathrel{\hstretch{0.7}{<\!\!\!=\!\!<}}\Varid{g})\;\Varid{x}\mathrel{=}\Varid{f}\mathbin{\hstretch{0.7}{=\!\!<\!\!<}}\Varid{g}\;\Varid{x}{}\<[E]%
\\[\blanklineskip]%
\>[B]{}(\mathrel{\raisebox{0.5\depth}{\scaleobj{0.5}{\langle}} \scaleobj{0.8}{\$} \raisebox{0.5\depth}{\scaleobj{0.5}{\rangle}}}){}\<[8]%
\>[8]{}\mathbin{::}(\Varid{a}\to \Varid{b})\to \Conid{M}\;\Varid{a}\to \Conid{M}\;\Varid{b}{}\<[E]%
\\
\>[B]{}\Varid{f}\mathrel{\raisebox{0.5\depth}{\scaleobj{0.5}{\langle}} \scaleobj{0.8}{\$} \raisebox{0.5\depth}{\scaleobj{0.5}{\rangle}}}\Varid{m}\mathrel{=}(\Varid{return}\mathbin{\cdot}\Varid{f})\mathbin{\hstretch{0.7}{=\!\!<\!\!<}}\Varid{m}{}\<[E]%
\\[\blanklineskip]%
\>[B]{}(\mathrel{\raisebox{0.5\depth}{\scaleobj{0.5}{\langle \bullet \rangle}}}){}\<[8]%
\>[8]{}\mathbin{::}(\Varid{b}\to \Varid{c})\to (\Varid{a}\to \Conid{M}\;\Varid{b})\to (\Varid{a}\to \Conid{M}\;\Varid{c}){}\<[E]%
\\
\>[B]{}\Varid{f}\mathrel{\raisebox{0.5\depth}{\scaleobj{0.5}{\langle \bullet \rangle}}}\Varid{g}\mathrel{=}(\Varid{return}\mathbin{\cdot}\Varid{f})\mathrel{\hstretch{0.7}{<\!\!\!=\!\!<}}\Varid{g}{}\<[E]%
\\
\>[B]{}~~{}\<[E]%
\ColumnHook
\end{hscode}\resethooks
\vspace{-1cm}
\caption{Some monadic operators we find handy for this paper.}
\label{figure:monadic-operators}
\end{figure}

More operators we find useful are given in Figure \ref{figure:monadic-operators}. Right-to-left Kleisli composition, denoted by \ensuremath{(\mathrel{\hstretch{0.7}{<\!\!\!=\!\!<}})}, composes two monadic operations \ensuremath{\Varid{a}\to \Conid{M}\;\Varid{b}} and \ensuremath{\Varid{b}\to \Conid{M}\;\Varid{c}} into an operation \ensuremath{\Varid{a}\to \Conid{M}\;\Varid{c}}. Operators \ensuremath{(\mathrel{\raisebox{0.5\depth}{\scaleobj{0.5}{\langle}} \scaleobj{0.8}{\$} \raisebox{0.5\depth}{\scaleobj{0.5}{\rangle}}})} and \ensuremath{(\mathrel{\raisebox{0.5\depth}{\scaleobj{0.5}{\langle \bullet \rangle}}})} are monadic counterparts of function application and composition: \ensuremath{(\mathrel{\raisebox{0.5\depth}{\scaleobj{0.5}{\langle}} \scaleobj{0.8}{\$} \raisebox{0.5\depth}{\scaleobj{0.5}{\rangle}}})} applies a pure function to a monad, while \ensuremath{(\mathrel{\raisebox{0.5\depth}{\scaleobj{0.5}{\langle \bullet \rangle}}})} composes a pure function after a monadic function.

We now introduce a collections of properties that allows us to rotate an expression that involves two operators and three operands.
These properties will be handy when we need to move parenthesis around in expressions.
To begin with, the following properties show that \ensuremath{(\mathrel{\raisebox{0.5\depth}{\scaleobj{0.5}{\langle}} \scaleobj{0.8}{\$} \raisebox{0.5\depth}{\scaleobj{0.5}{\rangle}}})} and \ensuremath{(\mathrel{\raisebox{0.5\depth}{\scaleobj{0.5}{\langle \bullet \rangle}}})} share properties similar to pure function application and composition:
\begin{align}
  \ensuremath{(\Varid{f}\mathrel{\raisebox{0.5\depth}{\scaleobj{0.5}{\langle \bullet \rangle}}}\Varid{g})\;\Varid{x}} &= \ensuremath{\Varid{f}\mathrel{\raisebox{0.5\depth}{\scaleobj{0.5}{\langle}} \scaleobj{0.8}{\$} \raisebox{0.5\depth}{\scaleobj{0.5}{\rangle}}}\Varid{g}\;\Varid{x}} \mbox{~~,}
    \label{eq:comp-ap}\\
  \ensuremath{\Varid{f}\mathrel{\raisebox{0.5\depth}{\scaleobj{0.5}{\langle}} \scaleobj{0.8}{\$} \raisebox{0.5\depth}{\scaleobj{0.5}{\rangle}}}(\Varid{g}\mathrel{\raisebox{0.5\depth}{\scaleobj{0.5}{\langle}} \scaleobj{0.8}{\$} \raisebox{0.5\depth}{\scaleobj{0.5}{\rangle}}}\Varid{m})} &= \ensuremath{(\Varid{f}\mathbin{\cdot}\Varid{g})\mathrel{\raisebox{0.5\depth}{\scaleobj{0.5}{\langle}} \scaleobj{0.8}{\$} \raisebox{0.5\depth}{\scaleobj{0.5}{\rangle}}}\Varid{m}} \mbox{~~,}
    \label{eq:comp-ap-ap}\\
  \ensuremath{\Varid{f}\mathrel{\raisebox{0.5\depth}{\scaleobj{0.5}{\langle \bullet \rangle}}}(\Varid{g}\mathrel{\raisebox{0.5\depth}{\scaleobj{0.5}{\langle \bullet \rangle}}}\Varid{h})} &= \ensuremath{(\Varid{f}\mathbin{\cdot}\Varid{g})\mathrel{\raisebox{0.5\depth}{\scaleobj{0.5}{\langle \bullet \rangle}}}\Varid{h}} \mbox{~~.}
    \label{eq:comp-mcomp-mcomp}
\end{align}
We also have the following law that allows us to rotate an expression that uses \ensuremath{(\mathrel{\raisebox{0.5\depth}{\scaleobj{0.5}{\langle \bullet \rangle}}})} and \ensuremath{(\mathbin{\cdot})}:
\begin{align}
  \ensuremath{\Varid{f}\mathrel{\raisebox{0.5\depth}{\scaleobj{0.5}{\langle \bullet \rangle}}}(\Varid{g}\mathbin{\cdot}\Varid{h})}   &= \ensuremath{(\Varid{f}\mathrel{\raisebox{0.5\depth}{\scaleobj{0.5}{\langle \bullet \rangle}}}\Varid{g})\mathbin{\cdot}\Varid{h}} \mbox{~~.}
      \label{eq:mcomp-comp-mcomp}
\end{align}
Note that \ensuremath{\Varid{g}} in~\eqref{eq:mcomp-comp-mcomp} must be a function returning a monad. Furthermore, \eqref{eq:comp-bind-ap} and~\eqref{eq:mcomp-bind-ap} relate \ensuremath{(\mathbin{\hstretch{0.7}{=\!\!<\!\!<}})} and \ensuremath{(\mathrel{\raisebox{0.5\depth}{\scaleobj{0.5}{\langle}} \scaleobj{0.8}{\$} \raisebox{0.5\depth}{\scaleobj{0.5}{\rangle}}})}, both operators applying functions to monads,
while~\eqref{eq:kc-mcomp} and~\eqref{eq:mcomp-kc} relate \ensuremath{(\mathrel{\hstretch{0.7}{<\!\!\!=\!\!<}})} and \ensuremath{(\mathrel{\raisebox{0.5\depth}{\scaleobj{0.5}{\langle \bullet \rangle}}})}, both operators composing functions on monads:
\begin{align}
  \ensuremath{\Varid{f}\mathbin{\hstretch{0.7}{=\!\!<\!\!<}}(\Varid{g}\mathrel{\raisebox{0.5\depth}{\scaleobj{0.5}{\langle}} \scaleobj{0.8}{\$} \raisebox{0.5\depth}{\scaleobj{0.5}{\rangle}}}\Varid{m})} &= \ensuremath{(\Varid{f}\mathbin{\cdot}\Varid{g})\mathbin{\hstretch{0.7}{=\!\!<\!\!<}}\Varid{m}} \mbox{~~,}
    \label{eq:comp-bind-ap}\\
  \ensuremath{\Varid{f}\mathrel{\raisebox{0.5\depth}{\scaleobj{0.5}{\langle}} \scaleobj{0.8}{\$} \raisebox{0.5\depth}{\scaleobj{0.5}{\rangle}}}(\Varid{g}\mathbin{\hstretch{0.7}{=\!\!<\!\!<}}\Varid{m})} &= \ensuremath{(\Varid{f}\mathrel{\raisebox{0.5\depth}{\scaleobj{0.5}{\langle \bullet \rangle}}}\Varid{g})\mathbin{\hstretch{0.7}{=\!\!<\!\!<}}\Varid{m}} \mbox{~~,}
    \label{eq:mcomp-bind-ap}\\
  \ensuremath{\Varid{f}\mathrel{\hstretch{0.7}{<\!\!\!=\!\!<}}(\Varid{g}\mathrel{\raisebox{0.5\depth}{\scaleobj{0.5}{\langle \bullet \rangle}}}\Varid{h})} &= \ensuremath{(\Varid{f}\mathbin{\cdot}\Varid{g})\mathrel{\hstretch{0.7}{<\!\!\!=\!\!<}}\Varid{h}} \mbox{~~,}
    \label{eq:kc-mcomp}\\
  \ensuremath{\Varid{f}\mathrel{\raisebox{0.5\depth}{\scaleobj{0.5}{\langle \bullet \rangle}}}(\Varid{g}\mathrel{\hstretch{0.7}{<\!\!\!=\!\!<}}\Varid{h})} &= \ensuremath{(\Varid{f}\mathrel{\raisebox{0.5\depth}{\scaleobj{0.5}{\langle \bullet \rangle}}}\Varid{g})\mathrel{\hstretch{0.7}{<\!\!\!=\!\!<}}\Varid{h}} \mbox{~~.}
    \label{eq:mcomp-kc}
\end{align}
Having these properties is one of the advantages of writing \ensuremath{(\mathbin{\hstretch{0.7}{=\!\!<\!\!<}})} and \ensuremath{(\mathrel{\hstretch{0.7}{<\!\!\!=\!\!<}})} backwards.
All the properties above can be proved by expanding definitions, and it is a good warming-up exercise proving some of them.
Some of them are proved in Appendix~\ref{sec:misc-proofs}.

None of these operators and properties are strictly necessary: they can all be reduced to \ensuremath{\Varid{return}}, \ensuremath{(\mathbin{\hstretch{0.7}{=\!\!<\!\!<}})}, and $\lambda$-abstractions. As is often the case when designing notations, having more operators allows ideas to be expressed concisely in a higher level of abstraction, at the expense of having more properties to memorise. It is personal preference where the balance should be. Properties \eqref{eq:comp-ap} through \eqref{eq:mcomp-kc} may look like a lot of properties to remember. In practice, we find it usually sufficient to let us be guided by types. For example, when we have \ensuremath{\Varid{f}\mathrel{\raisebox{0.5\depth}{\scaleobj{0.5}{\langle}} \scaleobj{0.8}{\$} \raisebox{0.5\depth}{\scaleobj{0.5}{\rangle}}}\Varid{g}\;\Varid{x}} and want to bring \ensuremath{\Varid{f}} and \ensuremath{\Varid{g}} together, by their types we can figure out the resulting expression should be \ensuremath{(\Varid{f}\mathrel{\raisebox{0.5\depth}{\scaleobj{0.5}{\langle \bullet \rangle}}}\Varid{g})\;\Varid{x}}.

\paragraph{Non-determinism Monad}

Non-determinism is the only effect we use in this report.
We assume two operators \ensuremath{\emptyset} and \ensuremath{(\talloblong)}: the former denotes failure, while \ensuremath{\Varid{m}\mathbin{\talloblong}\Varid{n}} denotes that the computation may yield either \ensuremath{\Varid{m}} or \ensuremath{\Varid{n}}.
As pointed out by \citet{GibbonsHinze:11:Just}, for proofs and derivations, what matters is not how a monad is implemented but what properties its operators satisfy.
What laws \ensuremath{\emptyset} and \ensuremath{(\talloblong)} should satisfy, however, can be a tricky issue.
As discussed by \citet{Kiselyov:15:Laws}, it eventually comes down to what we use the monad for.
It is usually expected that \ensuremath{(\Varid{a},(\talloblong),\emptyset)} be a monoid. That is, \ensuremath{(\talloblong)} is associative, with \ensuremath{\emptyset} as its zero:
\begin{align*}
\ensuremath{(\Varid{m}\mathbin{\talloblong}\Varid{n})\mathbin{\talloblong}\Varid{k}}~ &=~ \ensuremath{\Varid{m}\mathbin{\talloblong}(\Varid{n}\mathbin{\talloblong}\Varid{k})} \mbox{~~,}\\
\ensuremath{\emptyset\mathbin{\talloblong}\Varid{m}} ~=~ & \ensuremath{\Varid{m}} ~=~ \ensuremath{\Varid{m}\mathbin{\talloblong}\emptyset} \mbox{~~.}
\end{align*}
It is also assumed that monadic bind distributes into \ensuremath{(\talloblong)} from the end,
while \ensuremath{\emptyset} is a right zero for \ensuremath{(\mathbin{\hstretch{0.7}{=\!\!<\!\!<}})}:
\begin{align}
  \ensuremath{\Varid{f}\mathbin{\hstretch{0.7}{=\!\!<\!\!<}}(\Varid{m}_{1}\mathbin{\talloblong}\Varid{m}_{2})} ~&=~ \ensuremath{(\Varid{f}\mathbin{\hstretch{0.7}{=\!\!<\!\!<}}\Varid{m}_{1})\mathbin{\talloblong}(\Varid{f}\mathbin{\hstretch{0.7}{=\!\!<\!\!<}}\Varid{m}_{2})} \mbox{~~,}
  \label{eq:bind-mplus-dist}\\
  \ensuremath{\Varid{f}\mathbin{\hstretch{0.7}{=\!\!<\!\!<}}\emptyset} ~&=~ \ensuremath{\emptyset} \label{eq:bind-mzero-zero} \mbox{~~.}
\end{align}
For our purpose in this section, we also assume that \ensuremath{(\talloblong)} is commutative (\ensuremath{\Varid{m}\mathbin{\talloblong}\Varid{n}\mathrel{=}\Varid{n}\mathbin{\talloblong}\Varid{m}}) and idempotent (\ensuremath{\Varid{m}\mathbin{\talloblong}\Varid{m}\mathrel{=}\Varid{m}}). Implementation of such non-determinism monads have been studied by \citet{Fischer:11:Purely}.

Here are some induced laws about how \ensuremath{(\mathrel{\raisebox{0.5\depth}{\scaleobj{0.5}{\langle}} \scaleobj{0.8}{\$} \raisebox{0.5\depth}{\scaleobj{0.5}{\rangle}}})} interacts with \ensuremath{\Varid{return}} and non-determinism operators:
\begin{align}
  \ensuremath{\Varid{f}\mathrel{\raisebox{0.5\depth}{\scaleobj{0.5}{\langle}} \scaleobj{0.8}{\$} \raisebox{0.5\depth}{\scaleobj{0.5}{\rangle}}}\Varid{return}\;\Varid{x}} &= \ensuremath{\Varid{return}\;(\Varid{f}\;\Varid{x})} \mbox{~~,}\label{eq:ap-return}\\
  \ensuremath{\Varid{f}\mathrel{\raisebox{0.5\depth}{\scaleobj{0.5}{\langle}} \scaleobj{0.8}{\$} \raisebox{0.5\depth}{\scaleobj{0.5}{\rangle}}}\emptyset} &= \ensuremath{\emptyset} \mbox{~~,} \label{eq:ap-mzero}\\
  \ensuremath{\Varid{f}\mathrel{\raisebox{0.5\depth}{\scaleobj{0.5}{\langle}} \scaleobj{0.8}{\$} \raisebox{0.5\depth}{\scaleobj{0.5}{\rangle}}}(\Varid{m}_{1}\mathbin{\talloblong}\Varid{m}_{2})} &= \ensuremath{(\Varid{f}\mathrel{\raisebox{0.5\depth}{\scaleobj{0.5}{\langle}} \scaleobj{0.8}{\$} \raisebox{0.5\depth}{\scaleobj{0.5}{\rangle}}}\Varid{m}_{1})\mathbin{\talloblong}(\Varid{f}\mathrel{\raisebox{0.5\depth}{\scaleobj{0.5}{\langle}} \scaleobj{0.8}{\$} \raisebox{0.5\depth}{\scaleobj{0.5}{\rangle}}}\Varid{m}_{2})}\mbox{~~.}
      \label{eq:ap-mplus}
\end{align}
\section{Permutation and Insertion}
\label{sec:perm-insert}

As a warm-up example, the function \ensuremath{\Varid{perm}} non-deterministically computes a permutation of its input, using an auxiliary function \ensuremath{\Varid{insert}} that inserts an element to an arbitrary position in a list:
\begin{hscode}\SaveRestoreHook
\column{B}{@{}>{\hspre}l<{\hspost}@{}}%
\column{18}{@{}>{\hspre}c<{\hspost}@{}}%
\column{18E}{@{}l@{}}%
\column{22}{@{}>{\hspre}l<{\hspost}@{}}%
\column{E}{@{}>{\hspre}l<{\hspost}@{}}%
\>[B]{}\Varid{perm}{}\<[18]%
\>[18]{}\mathbin{::}{}\<[18E]%
\>[22]{}[\mskip1.5mu \Varid{a}\mskip1.5mu]\to \Conid{M}\;[\mskip1.5mu \Varid{a}\mskip1.5mu]{}\<[E]%
\\
\>[B]{}\Varid{perm}\;[\mskip1.5mu \mskip1.5mu]{}\<[18]%
\>[18]{}\mathrel{=}{}\<[18E]%
\>[22]{}\Varid{return}\;[\mskip1.5mu \mskip1.5mu]{}\<[E]%
\\
\>[B]{}\Varid{perm}\;(\Varid{x}\mathbin{:}\Varid{xs}){}\<[18]%
\>[18]{}\mathrel{=}{}\<[18E]%
\>[22]{}\Varid{insert}\;\Varid{x}\mathbin{\hstretch{0.7}{=\!\!<\!\!<}}\Varid{perm}\;\Varid{xs}~~,{}\<[E]%
\\[\blanklineskip]%
\>[B]{}\Varid{insert}{}\<[18]%
\>[18]{}\mathbin{::}{}\<[18E]%
\>[22]{}\Varid{a}\to [\mskip1.5mu \Varid{a}\mskip1.5mu]\to \Conid{M}\;[\mskip1.5mu \Varid{a}\mskip1.5mu]{}\<[E]%
\\
\>[B]{}\Varid{insert}\;\Varid{x}\;[\mskip1.5mu \mskip1.5mu]{}\<[18]%
\>[18]{}\mathrel{=}{}\<[18E]%
\>[22]{}\Varid{return}\;[\mskip1.5mu \Varid{x}\mskip1.5mu]{}\<[E]%
\\
\>[B]{}\Varid{insert}\;\Varid{x}\;(\Varid{y}\mathbin{:}\Varid{xs}){}\<[18]%
\>[18]{}\mathrel{=}{}\<[18E]%
\>[22]{}\Varid{return}\;(\Varid{x}\mathbin{:}\Varid{y}\mathbin{:}\Varid{xs})\mathbin{\talloblong}((\Varid{y}\mathbin{:})\mathrel{\raisebox{0.5\depth}{\scaleobj{0.5}{\langle}} \scaleobj{0.8}{\$} \raisebox{0.5\depth}{\scaleobj{0.5}{\rangle}}}\Varid{insert}\;\Varid{x}\;\Varid{xs})~~.{}\<[E]%
\ColumnHook
\end{hscode}\resethooks
For example, possible results of \ensuremath{\Varid{perm}\;[\mskip1.5mu \mathrm{0},\mathrm{1},\mathrm{2}\mskip1.5mu]} include
\ensuremath{[\mskip1.5mu \mathrm{0},\mathrm{1},\mathrm{2}\mskip1.5mu]}, \ensuremath{[\mskip1.5mu \mathrm{0},\mathrm{2},\mathrm{1}\mskip1.5mu]}, \ensuremath{[\mskip1.5mu \mathrm{1},\mathrm{0},\mathrm{2}\mskip1.5mu]}, \ensuremath{[\mskip1.5mu \mathrm{1},\mathrm{2},\mathrm{0}\mskip1.5mu]}, \ensuremath{[\mskip1.5mu \mathrm{2},\mathrm{0},\mathrm{1}\mskip1.5mu]}, and
\ensuremath{[\mskip1.5mu \mathrm{2},\mathrm{1},\mathrm{0}\mskip1.5mu]}.

\paragraph{Determinism}
The following lemma presents properties under which permuting the input list does not change the result of a \ensuremath{\Varid{foldr}}:
\begin{lemma} \label{lemma:fold-perm}
Given \ensuremath{(\odot)\mathbin{::}\Varid{a}\to \Varid{b}\to \Varid{b}}. If \ensuremath{\Varid{x}\mathbin{\odot}(\Varid{y}\mathbin{\odot}\Varid{z})\mathrel{=}\Varid{y}\mathbin{\odot}(\Varid{x}\mathbin{\odot}\Varid{z})} for all \ensuremath{\Varid{x},\Varid{y}\mathbin{::}\Varid{a}} and \ensuremath{\Varid{z}\mathbin{::}\Varid{b}}, we have\begin{hscode}\SaveRestoreHook
\column{B}{@{}>{\hspre}l<{\hspost}@{}}%
\column{3}{@{}>{\hspre}l<{\hspost}@{}}%
\column{E}{@{}>{\hspre}l<{\hspost}@{}}%
\>[3]{}\Varid{foldr}\;(\odot)\;\Varid{z}\mathrel{\raisebox{0.5\depth}{\scaleobj{0.5}{\langle \bullet \rangle}}}\Varid{perm}\mathrel{=}\Varid{return}\mathbin{\cdot}\Varid{foldr}\;(\odot)\;\Varid{z}~~.{}\<[E]%
\ColumnHook
\end{hscode}\resethooks
\end{lemma}

Since \ensuremath{\Varid{perm}} is defined in terms of \ensuremath{\Varid{insert}}, proof of Lemma~\ref{lemma:fold-perm} naturally depends on a lemma about a related property of \ensuremath{\Varid{insert}}:
\begin{lemma} \label{lemma:fold-insert}
Given \ensuremath{(\odot)\mathbin{::}\Varid{a}\to \Varid{b}\to \Varid{b}}, we have\begin{hscode}\SaveRestoreHook
\column{B}{@{}>{\hspre}l<{\hspost}@{}}%
\column{3}{@{}>{\hspre}l<{\hspost}@{}}%
\column{E}{@{}>{\hspre}l<{\hspost}@{}}%
\>[3]{}\Varid{foldr}\;(\odot)\;\Varid{z}\mathrel{\raisebox{0.5\depth}{\scaleobj{0.5}{\langle \bullet \rangle}}}\Varid{insert}\;\Varid{x}\mathrel{=}\Varid{return}\mathbin{\cdot}\Varid{foldr}\;(\odot)\;\Varid{z}\mathbin{\cdot}(\Varid{x}\mathbin{:})~~,{}\<[E]%
\ColumnHook
\end{hscode}\resethooks
provided that \ensuremath{\Varid{x}\mathbin{\odot}(\Varid{y}\mathbin{\odot}\Varid{z})\mathrel{=}\Varid{y}\mathbin{\odot}(\Varid{x}\mathbin{\odot}\Varid{z})}
for all \ensuremath{\Varid{x},\Varid{y}\mathbin{::}\Varid{a}} and \ensuremath{\Varid{z}\mathbin{::}\Varid{b}}.
\end{lemma}
\begin{proof} Prove \ensuremath{\Varid{foldr}\;(\odot)\;\Varid{z}\mathrel{\raisebox{0.5\depth}{\scaleobj{0.5}{\langle}} \scaleobj{0.8}{\$} \raisebox{0.5\depth}{\scaleobj{0.5}{\rangle}}}\Varid{insert}\;\Varid{x}\;\Varid{xs}\mathrel{=}\Varid{return}\;(\Varid{foldr}\;(\odot)\;\Varid{z}\;(\Varid{x}\mathbin{:}\Varid{xs}))}. Induction on \ensuremath{\Varid{xs}}.

{\sc Case} \ensuremath{\Varid{xs}\mathbin{:=}[\mskip1.5mu \mskip1.5mu]}:
\begin{hscode}\SaveRestoreHook
\column{B}{@{}>{\hspre}l<{\hspost}@{}}%
\column{4}{@{}>{\hspre}l<{\hspost}@{}}%
\column{7}{@{}>{\hspre}l<{\hspost}@{}}%
\column{31}{@{}>{\hspre}l<{\hspost}@{}}%
\column{E}{@{}>{\hspre}l<{\hspost}@{}}%
\>[4]{}\Varid{foldr}\;(\odot)\;\Varid{z}\mathrel{\raisebox{0.5\depth}{\scaleobj{0.5}{\langle}} \scaleobj{0.8}{\$} \raisebox{0.5\depth}{\scaleobj{0.5}{\rangle}}}\Varid{insert}\;\Varid{x}\;[\mskip1.5mu \mskip1.5mu]{}\<[E]%
\\
\>[B]{}\mathbin{=}{}\<[7]%
\>[7]{}\mbox{\commentbegin  definition of \ensuremath{\Varid{insert}}  \commentend}{}\<[E]%
\\
\>[B]{}\hsindent{4}{}\<[4]%
\>[4]{}\Varid{foldr}\;(\odot)\;\Varid{z}\mathrel{\raisebox{0.5\depth}{\scaleobj{0.5}{\langle}} \scaleobj{0.8}{\$} \raisebox{0.5\depth}{\scaleobj{0.5}{\rangle}}}\Varid{return}\;[\mskip1.5mu \Varid{x}\mskip1.5mu]{}\<[E]%
\\
\>[B]{}\mathbin{=}{}\<[7]%
\>[7]{}\mbox{\commentbegin  by \eqref{eq:ap-return}  \commentend}{}\<[E]%
\\
\>[B]{}\hsindent{4}{}\<[4]%
\>[4]{}\Varid{return}\;(\Varid{foldr}\;(\odot)\;\Varid{z}\;[\mskip1.5mu \Varid{x}\mskip1.5mu]){}\<[31]%
\>[31]{}~~.{}\<[E]%
\ColumnHook
\end{hscode}\resethooks

{\sc Case} \ensuremath{\Varid{xs}\mathbin{:=}\Varid{y}\mathbin{:}\Varid{xs}}:
\begin{hscode}\SaveRestoreHook
\column{B}{@{}>{\hspre}l<{\hspost}@{}}%
\column{4}{@{}>{\hspre}l<{\hspost}@{}}%
\column{9}{@{}>{\hspre}l<{\hspost}@{}}%
\column{E}{@{}>{\hspre}l<{\hspost}@{}}%
\>[4]{}\Varid{foldr}\;(\odot)\;\Varid{z}\mathrel{\raisebox{0.5\depth}{\scaleobj{0.5}{\langle}} \scaleobj{0.8}{\$} \raisebox{0.5\depth}{\scaleobj{0.5}{\rangle}}}\Varid{insert}\;\Varid{x}\;(\Varid{y}\mathbin{:}\Varid{xs}){}\<[E]%
\\
\>[B]{}\mathbin{=}{}\<[9]%
\>[9]{}\mbox{\commentbegin  definition of \ensuremath{\Varid{insert}}  \commentend}{}\<[E]%
\\
\>[B]{}\hsindent{4}{}\<[4]%
\>[4]{}\Varid{foldr}\;(\odot)\;\Varid{z}\mathrel{\raisebox{0.5\depth}{\scaleobj{0.5}{\langle}} \scaleobj{0.8}{\$} \raisebox{0.5\depth}{\scaleobj{0.5}{\rangle}}}(\Varid{return}\;(\Varid{x}\mathbin{:}\Varid{y}\mathbin{:}\Varid{xs})\mathbin{\talloblong}((\Varid{y}\mathbin{:})\mathrel{\raisebox{0.5\depth}{\scaleobj{0.5}{\langle}} \scaleobj{0.8}{\$} \raisebox{0.5\depth}{\scaleobj{0.5}{\rangle}}}\Varid{insert}\;\Varid{x}\;\Varid{xs})){}\<[E]%
\\
\>[B]{}\mathbin{=}{}\<[9]%
\>[9]{}\mbox{\commentbegin  by \eqref{eq:ap-mplus}, \eqref{eq:ap-return}, and \eqref{eq:comp-ap-ap}  \commentend}{}\<[E]%
\\
\>[B]{}\hsindent{4}{}\<[4]%
\>[4]{}\Varid{return}\;(\Varid{foldr}\;(\odot)\;\Varid{z}\;(\Varid{x}\mathbin{:}\Varid{y}\mathbin{:}\Varid{xs}))\mathbin{\talloblong}((\Varid{foldr}\;(\odot)\;\Varid{z}\mathbin{\cdot}(\Varid{y}\mathbin{:}))\mathrel{\raisebox{0.5\depth}{\scaleobj{0.5}{\langle}} \scaleobj{0.8}{\$} \raisebox{0.5\depth}{\scaleobj{0.5}{\rangle}}}\Varid{insert}\;\Varid{x}\;\Varid{xs})~~.{}\<[E]%
\ColumnHook
\end{hscode}\resethooks
Focus on the second branch of \ensuremath{(\talloblong)}:
\begin{hscode}\SaveRestoreHook
\column{B}{@{}>{\hspre}l<{\hspost}@{}}%
\column{4}{@{}>{\hspre}l<{\hspost}@{}}%
\column{9}{@{}>{\hspre}l<{\hspost}@{}}%
\column{10}{@{}>{\hspre}l<{\hspost}@{}}%
\column{E}{@{}>{\hspre}l<{\hspost}@{}}%
\>[4]{}(\Varid{foldr}\;(\odot)\;\Varid{z}\mathbin{\cdot}(\Varid{y}\mathbin{:}))\mathrel{\raisebox{0.5\depth}{\scaleobj{0.5}{\langle}} \scaleobj{0.8}{\$} \raisebox{0.5\depth}{\scaleobj{0.5}{\rangle}}}\Varid{insert}\;\Varid{x}\;\Varid{xs}{}\<[E]%
\\
\>[B]{}\mathbin{=}{}\<[9]%
\>[9]{}\mbox{\commentbegin  definition of \ensuremath{\Varid{foldr}}  \commentend}{}\<[E]%
\\
\>[B]{}\hsindent{4}{}\<[4]%
\>[4]{}((\Varid{y}\mathbin{\odot})\mathbin{\cdot}\Varid{foldr}\;(\odot)\;\Varid{z})\mathrel{\raisebox{0.5\depth}{\scaleobj{0.5}{\langle}} \scaleobj{0.8}{\$} \raisebox{0.5\depth}{\scaleobj{0.5}{\rangle}}}\Varid{insert}\;\Varid{x}\;\Varid{xs}{}\<[E]%
\\
\>[B]{}\mathbin{=}{}\<[9]%
\>[9]{}\mbox{\commentbegin  by \eqref{eq:comp-ap-ap}  \commentend}{}\<[E]%
\\
\>[B]{}\hsindent{4}{}\<[4]%
\>[4]{}(\Varid{y}\mathbin{\odot})\mathrel{\raisebox{0.5\depth}{\scaleobj{0.5}{\langle}} \scaleobj{0.8}{\$} \raisebox{0.5\depth}{\scaleobj{0.5}{\rangle}}}(\Varid{foldr}\;(\odot)\;\Varid{z}\mathrel{\raisebox{0.5\depth}{\scaleobj{0.5}{\langle}} \scaleobj{0.8}{\$} \raisebox{0.5\depth}{\scaleobj{0.5}{\rangle}}}\Varid{insert}\;\Varid{x}\;\Varid{xs}){}\<[E]%
\\
\>[B]{}\mathbin{=}{}\<[10]%
\>[10]{}\mbox{\commentbegin  induction  \commentend}{}\<[E]%
\\
\>[B]{}\hsindent{4}{}\<[4]%
\>[4]{}(\Varid{y}\mathbin{\odot})\mathrel{\raisebox{0.5\depth}{\scaleobj{0.5}{\langle}} \scaleobj{0.8}{\$} \raisebox{0.5\depth}{\scaleobj{0.5}{\rangle}}}\Varid{return}\;(\Varid{foldr}\;(\odot)\;\Varid{z}\;(\Varid{x}\mathbin{:}\Varid{xs})){}\<[E]%
\\
\>[B]{}\mathbin{=}{}\<[10]%
\>[10]{}\mbox{\commentbegin  by \eqref{eq:ap-return}  \commentend}{}\<[E]%
\\
\>[B]{}\hsindent{4}{}\<[4]%
\>[4]{}\Varid{return}\;(\Varid{y}\mathbin{\odot}\Varid{foldr}\;(\odot)\;\Varid{z}\;(\Varid{x}\mathbin{:}\Varid{xs})){}\<[E]%
\\
\>[B]{}\mathbin{=}{}\<[10]%
\>[10]{}\mbox{\commentbegin  definition of \ensuremath{\Varid{foldr}}  \commentend}{}\<[E]%
\\
\>[B]{}\hsindent{4}{}\<[4]%
\>[4]{}\Varid{return}\;(\Varid{y}\mathbin{\odot}(\Varid{x}\mathbin{\odot}\Varid{foldr}\;(\odot)\;\Varid{z}\;\Varid{xs})){}\<[E]%
\\
\>[B]{}\mathbin{=}{}\<[10]%
\>[10]{}\mbox{\commentbegin  since \ensuremath{\Varid{x}\mathbin{\odot}(\Varid{y}\mathbin{\odot}\Varid{z})\mathrel{=}\Varid{y}\mathbin{\odot}(\Varid{x}\mathbin{\odot}\Varid{z})}  \commentend}{}\<[E]%
\\
\>[B]{}\hsindent{4}{}\<[4]%
\>[4]{}\Varid{return}\;(\Varid{foldr}\;(\odot)\;\Varid{z}\;(\Varid{x}\mathbin{:}\Varid{y}\mathbin{:}\Varid{xs}))~~.{}\<[E]%
\ColumnHook
\end{hscode}\resethooks
Thus we have
\begin{hscode}\SaveRestoreHook
\column{B}{@{}>{\hspre}l<{\hspost}@{}}%
\column{4}{@{}>{\hspre}l<{\hspost}@{}}%
\column{9}{@{}>{\hspre}l<{\hspost}@{}}%
\column{E}{@{}>{\hspre}l<{\hspost}@{}}%
\>[4]{}(\Varid{foldr}\;(\odot)\;\Varid{z}\mathrel{\raisebox{0.5\depth}{\scaleobj{0.5}{\langle \bullet \rangle}}}\Varid{insert}\;\Varid{x})\;(\Varid{y}\mathbin{:}\Varid{xs}){}\<[E]%
\\
\>[B]{}\mathbin{=}{}\<[9]%
\>[9]{}\mbox{\commentbegin  calculation above  \commentend}{}\<[E]%
\\
\>[B]{}\hsindent{4}{}\<[4]%
\>[4]{}\Varid{return}\;(\Varid{foldr}\;(\odot)\;\Varid{z}\;(\Varid{x}\mathbin{:}\Varid{y}\mathbin{:}\Varid{xs}))\mathbin{\talloblong}\Varid{return}\;(\Varid{foldr}\;(\odot)\;\Varid{z}\;(\Varid{x}\mathbin{:}\Varid{y}\mathbin{:}\Varid{xs})){}\<[E]%
\\
\>[B]{}\mathbin{=}{}\<[9]%
\>[9]{}\mbox{\commentbegin  idempotence of \ensuremath{(\talloblong)}  \commentend}{}\<[E]%
\\
\>[B]{}\hsindent{4}{}\<[4]%
\>[4]{}\Varid{return}\;(\Varid{foldr}\;(\odot)\;\Varid{z}\;(\Varid{x}\mathbin{:}\Varid{y}\mathbin{:}\Varid{xs}))~~.{}\<[E]%
\ColumnHook
\end{hscode}\resethooks
\end{proof}

Proof of Lemma~\ref{lemma:fold-perm} then follows:
\begin{proof}
Prove that \ensuremath{\Varid{foldr}\;(\odot)\;\Varid{z}\mathrel{\raisebox{0.5\depth}{\scaleobj{0.5}{\langle}} \scaleobj{0.8}{\$} \raisebox{0.5\depth}{\scaleobj{0.5}{\rangle}}}\Varid{perm}\;\Varid{xs}\mathrel{=}\Varid{return}\;(\Varid{foldr}\;(\odot)\;\Varid{z}\;\Varid{xs})}.
Induction on \ensuremath{\Varid{xs}}.

{\sc Case} \ensuremath{\Varid{xs}\mathbin{:=}[\mskip1.5mu \mskip1.5mu]}:
\begin{hscode}\SaveRestoreHook
\column{B}{@{}>{\hspre}l<{\hspost}@{}}%
\column{3}{@{}>{\hspre}l<{\hspost}@{}}%
\column{8}{@{}>{\hspre}l<{\hspost}@{}}%
\column{E}{@{}>{\hspre}l<{\hspost}@{}}%
\>[3]{}\Varid{foldr}\;(\odot)\;\Varid{z}\mathrel{\raisebox{0.5\depth}{\scaleobj{0.5}{\langle}} \scaleobj{0.8}{\$} \raisebox{0.5\depth}{\scaleobj{0.5}{\rangle}}}\Varid{perm}\;[\mskip1.5mu \mskip1.5mu]{}\<[E]%
\\
\>[B]{}\mathbin{=}{}\<[8]%
\>[8]{}\mbox{\commentbegin  definitions of \ensuremath{\Varid{perm}}  \commentend}{}\<[E]%
\\
\>[B]{}\hsindent{3}{}\<[3]%
\>[3]{}\Varid{foldr}\;(\odot)\;\Varid{z}\mathrel{\raisebox{0.5\depth}{\scaleobj{0.5}{\langle}} \scaleobj{0.8}{\$} \raisebox{0.5\depth}{\scaleobj{0.5}{\rangle}}}\Varid{return}\;[\mskip1.5mu \mskip1.5mu]{}\<[E]%
\\
\>[B]{}\mathbin{=}{}\<[8]%
\>[8]{}\mbox{\commentbegin  by \eqref{eq:ap-return}  \commentend}{}\<[E]%
\\
\>[B]{}\hsindent{3}{}\<[3]%
\>[3]{}\Varid{return}\;(\Varid{foldr}\;(\odot)\;\Varid{z}\;[\mskip1.5mu \mskip1.5mu])~~.{}\<[E]%
\ColumnHook
\end{hscode}\resethooks

{\sc Case} \ensuremath{\Varid{xs}\mathbin{:=}\Varid{x}\mathbin{:}\Varid{xs}}:
\begin{hscode}\SaveRestoreHook
\column{B}{@{}>{\hspre}l<{\hspost}@{}}%
\column{4}{@{}>{\hspre}l<{\hspost}@{}}%
\column{9}{@{}>{\hspre}l<{\hspost}@{}}%
\column{E}{@{}>{\hspre}l<{\hspost}@{}}%
\>[4]{}\Varid{foldr}\;(\odot)\;\Varid{z}\mathrel{\raisebox{0.5\depth}{\scaleobj{0.5}{\langle}} \scaleobj{0.8}{\$} \raisebox{0.5\depth}{\scaleobj{0.5}{\rangle}}}\Varid{perm}\;(\Varid{x}\mathbin{:}\Varid{xs}){}\<[E]%
\\
\>[B]{}\mathbin{=}{}\<[9]%
\>[9]{}\mbox{\commentbegin  definition of \ensuremath{\Varid{perm}}  \commentend}{}\<[E]%
\\
\>[B]{}\hsindent{4}{}\<[4]%
\>[4]{}\Varid{foldr}\;(\odot)\;\Varid{z}\mathrel{\raisebox{0.5\depth}{\scaleobj{0.5}{\langle}} \scaleobj{0.8}{\$} \raisebox{0.5\depth}{\scaleobj{0.5}{\rangle}}}(\Varid{insert}\;\Varid{x}\mathbin{\hstretch{0.7}{=\!\!<\!\!<}}\Varid{perm}\;\Varid{xs}){}\<[E]%
\\
\>[B]{}\mathbin{=}{}\<[9]%
\>[9]{}\mbox{\commentbegin  by~\eqref{eq:mcomp-bind-ap}  \commentend}{}\<[E]%
\\
\>[B]{}\hsindent{4}{}\<[4]%
\>[4]{}(\Varid{foldr}\;(\odot)\;\Varid{z}\mathrel{\raisebox{0.5\depth}{\scaleobj{0.5}{\langle \bullet \rangle}}}\Varid{insert}\;\Varid{x})\mathbin{\hstretch{0.7}{=\!\!<\!\!<}}\Varid{perm}\;\Varid{xs}{}\<[E]%
\\
\>[B]{}\mathbin{=}{}\<[9]%
\>[9]{}\mbox{\commentbegin  Lemma \ref{lemma:fold-insert} \commentend}{}\<[E]%
\\
\>[B]{}\hsindent{4}{}\<[4]%
\>[4]{}(\Varid{return}\mathbin{\cdot}\Varid{foldr}\;(\odot)\;\Varid{z}\mathbin{\cdot}(\Varid{x}\mathbin{:}))\mathbin{\hstretch{0.7}{=\!\!<\!\!<}}\Varid{perm}\;\Varid{xs}{}\<[E]%
\\
\>[B]{}\mathbin{=}{}\<[9]%
\>[9]{}\mbox{\commentbegin  definitions of \ensuremath{\Varid{foldr}} and \ensuremath{(\mathrel{\raisebox{0.5\depth}{\scaleobj{0.5}{\langle}} \scaleobj{0.8}{\$} \raisebox{0.5\depth}{\scaleobj{0.5}{\rangle}}})}  \commentend}{}\<[E]%
\\
\>[B]{}\hsindent{4}{}\<[4]%
\>[4]{}((\Varid{x}\mathbin{\odot})\mathbin{\cdot}\Varid{foldr}\;(\odot)\;\Varid{z})\mathrel{\raisebox{0.5\depth}{\scaleobj{0.5}{\langle}} \scaleobj{0.8}{\$} \raisebox{0.5\depth}{\scaleobj{0.5}{\rangle}}}\Varid{perm}\;\Varid{xs}{}\<[E]%
\\
\>[B]{}\mathbin{=}{}\<[9]%
\>[9]{}\mbox{\commentbegin  by \eqref{eq:comp-ap-ap}  \commentend}{}\<[E]%
\\
\>[B]{}\hsindent{4}{}\<[4]%
\>[4]{}(\Varid{x}\mathbin{\odot})\mathrel{\raisebox{0.5\depth}{\scaleobj{0.5}{\langle}} \scaleobj{0.8}{\$} \raisebox{0.5\depth}{\scaleobj{0.5}{\rangle}}}(\Varid{foldr}\;(\odot)\;\Varid{z}\mathrel{\raisebox{0.5\depth}{\scaleobj{0.5}{\langle}} \scaleobj{0.8}{\$} \raisebox{0.5\depth}{\scaleobj{0.5}{\rangle}}}\Varid{perm}\;\Varid{xs}){}\<[E]%
\\
\>[B]{}\mathbin{=}{}\<[9]%
\>[9]{}\mbox{\commentbegin  induction  \commentend}{}\<[E]%
\\
\>[B]{}\hsindent{4}{}\<[4]%
\>[4]{}(\Varid{x}\mathbin{\odot})\mathrel{\raisebox{0.5\depth}{\scaleobj{0.5}{\langle}} \scaleobj{0.8}{\$} \raisebox{0.5\depth}{\scaleobj{0.5}{\rangle}}}(\Varid{return}\;(\Varid{foldr}\;(\odot)\;\Varid{z}\;\Varid{xs})){}\<[E]%
\\
\>[B]{}\mathbin{=}{}\<[9]%
\>[9]{}\mbox{\commentbegin  by \eqref{eq:ap-return}  \commentend}{}\<[E]%
\\
\>[B]{}\hsindent{4}{}\<[4]%
\>[4]{}\Varid{return}\;(\Varid{x}\mathbin{\odot}\Varid{foldr}\;(\odot)\;\Varid{z}\;\Varid{xs}){}\<[E]%
\\
\>[B]{}\mathbin{=}{}\<[9]%
\>[9]{}\mbox{\commentbegin  definition of \ensuremath{\Varid{foldr}}  \commentend}{}\<[E]%
\\
\>[B]{}\hsindent{4}{}\<[4]%
\>[4]{}\Varid{return}\;(\Varid{foldr}\;(\odot)\;\Varid{z}\;(\Varid{x}\mathbin{:}\Varid{xs}))~~.{}\<[E]%
\ColumnHook
\end{hscode}\resethooks
\end{proof}

\paragraph{Map, Filter, and Permutation}
It is not hard for one to formulate the following relationship between
\ensuremath{\Varid{map}} and \ensuremath{\Varid{perm}}, which is also based on a related property relating \ensuremath{\Varid{map}} and \ensuremath{\Varid{insert}}:%
\footnote{Lemma \ref{lemma:shuffle-map} and \ref{lemma:insert-map} are in fact free theorems
of \ensuremath{\Varid{perm}} and \ensuremath{\Varid{insert}}~\citep{Voigtlander:09:Free}. They serve as good exercises, nevertheless.}
\begin{lemma}\label{lemma:shuffle-map}
  \ensuremath{\Varid{perm}\mathbin{\cdot}\Varid{map}\;\Varid{f}\mathrel{=}\Varid{map}\;\Varid{f}\mathrel{\raisebox{0.5\depth}{\scaleobj{0.5}{\langle \bullet \rangle}}}\Varid{perm}}.
\end{lemma}
\begin{lemma}\label{lemma:insert-map}
  \ensuremath{\Varid{insert}\;(\Varid{f}\;\Varid{x})\mathbin{\cdot}\Varid{map}\;\Varid{f}\mathrel{=}\Varid{map}\;\Varid{f}\mathrel{\raisebox{0.5\depth}{\scaleobj{0.5}{\langle \bullet \rangle}}}\Varid{insert}\;\Varid{x}}.
\end{lemma}
The lemma is true because \ensuremath{\Varid{map}\;\Varid{f}} is a pure computation --- in reasoning about monadic programs it is helpful, and sometimes essential, to identify its pure segments, because these are the parts more properties are applicable. Note that the composition \ensuremath{(\mathbin{\cdot})} on the lefthand side is turned into \ensuremath{(\mathrel{\raisebox{0.5\depth}{\scaleobj{0.5}{\langle \bullet \rangle}}})} once we move \ensuremath{\Varid{map}\;\Varid{f}} leftwards.

We prove only Lemma~\ref{lemma:insert-map}.
\begin{proof}
  Prove by induction on \ensuremath{\Varid{xs}} that \ensuremath{\Varid{map}\;\Varid{f}\mathrel{\raisebox{0.5\depth}{\scaleobj{0.5}{\langle}} \scaleobj{0.8}{\$} \raisebox{0.5\depth}{\scaleobj{0.5}{\rangle}}}\Varid{insert}\;\Varid{x}\;\Varid{xs}\mathrel{=}\Varid{insert}\;(\Varid{f}\;\Varid{x})\;(\Varid{map}\;\Varid{f}\;\Varid{xs})} for all \ensuremath{\Varid{xs}}. We present only the inductive case
\ensuremath{\Varid{xs}\mathbin{:=}\Varid{y}\mathbin{:}\Varid{xs}}:
\begin{hscode}\SaveRestoreHook
\column{B}{@{}>{\hspre}l<{\hspost}@{}}%
\column{4}{@{}>{\hspre}l<{\hspost}@{}}%
\column{8}{@{}>{\hspre}l<{\hspost}@{}}%
\column{E}{@{}>{\hspre}l<{\hspost}@{}}%
\>[4]{}\Varid{map}\;\Varid{f}\mathrel{\raisebox{0.5\depth}{\scaleobj{0.5}{\langle}} \scaleobj{0.8}{\$} \raisebox{0.5\depth}{\scaleobj{0.5}{\rangle}}}\Varid{insert}\;\Varid{x}\;(\Varid{y}\mathbin{:}\Varid{xs}){}\<[E]%
\\
\>[B]{}\mathbin{=}{}\<[8]%
\>[8]{}\mbox{\commentbegin  definition of \ensuremath{\Varid{insert}}  \commentend}{}\<[E]%
\\
\>[B]{}\hsindent{4}{}\<[4]%
\>[4]{}\Varid{map}\;\Varid{f}\mathrel{\raisebox{0.5\depth}{\scaleobj{0.5}{\langle}} \scaleobj{0.8}{\$} \raisebox{0.5\depth}{\scaleobj{0.5}{\rangle}}}(\Varid{return}\;(\Varid{x}\mathbin{:}\Varid{y}\mathbin{:}\Varid{xs})\mathbin{\talloblong}((\Varid{y}\mathbin{:})\mathrel{\raisebox{0.5\depth}{\scaleobj{0.5}{\langle}} \scaleobj{0.8}{\$} \raisebox{0.5\depth}{\scaleobj{0.5}{\rangle}}}\Varid{insert}\;\Varid{x}\;\Varid{xs})){}\<[E]%
\\
\>[B]{}\mathbin{=}{}\<[8]%
\>[8]{}\mbox{\commentbegin  by \eqref{eq:ap-mplus} and \eqref{eq:ap-return}  \commentend}{}\<[E]%
\\
\>[B]{}\hsindent{4}{}\<[4]%
\>[4]{}\Varid{return}\;(\Varid{map}\;\Varid{f}\;(\Varid{x}\mathbin{:}\Varid{y}\mathbin{:}\Varid{xs}))\mathbin{\talloblong}(\Varid{map}\;\Varid{f}\mathrel{\raisebox{0.5\depth}{\scaleobj{0.5}{\langle}} \scaleobj{0.8}{\$} \raisebox{0.5\depth}{\scaleobj{0.5}{\rangle}}}((\Varid{y}\mathbin{:})\mathrel{\raisebox{0.5\depth}{\scaleobj{0.5}{\langle}} \scaleobj{0.8}{\$} \raisebox{0.5\depth}{\scaleobj{0.5}{\rangle}}}\Varid{insert}\;\Varid{x}\;\Varid{xs}))~~.{}\<[E]%
\ColumnHook
\end{hscode}\resethooks
For the second branch we reason:
\begin{hscode}\SaveRestoreHook
\column{B}{@{}>{\hspre}l<{\hspost}@{}}%
\column{4}{@{}>{\hspre}l<{\hspost}@{}}%
\column{8}{@{}>{\hspre}l<{\hspost}@{}}%
\column{E}{@{}>{\hspre}l<{\hspost}@{}}%
\>[4]{}\Varid{map}\;\Varid{f}\mathrel{\raisebox{0.5\depth}{\scaleobj{0.5}{\langle}} \scaleobj{0.8}{\$} \raisebox{0.5\depth}{\scaleobj{0.5}{\rangle}}}((\Varid{y}\mathbin{:})\mathrel{\raisebox{0.5\depth}{\scaleobj{0.5}{\langle}} \scaleobj{0.8}{\$} \raisebox{0.5\depth}{\scaleobj{0.5}{\rangle}}}\Varid{insert}\;\Varid{x}\;\Varid{xs}){}\<[E]%
\\
\>[B]{}\mathbin{=}{}\<[8]%
\>[8]{}\mbox{\commentbegin  by \eqref{eq:comp-ap-ap}  \commentend}{}\<[E]%
\\
\>[B]{}\hsindent{4}{}\<[4]%
\>[4]{}(\Varid{map}\;\Varid{f}\mathbin{\cdot}(\Varid{y}\mathbin{:}))\mathrel{\raisebox{0.5\depth}{\scaleobj{0.5}{\langle}} \scaleobj{0.8}{\$} \raisebox{0.5\depth}{\scaleobj{0.5}{\rangle}}}\Varid{insert}\;\Varid{x}\;\Varid{xs}{}\<[E]%
\\
\>[B]{}\mathbin{=}{}\<[8]%
\>[8]{}\mbox{\commentbegin  definition of \ensuremath{\Varid{map}}  \commentend}{}\<[E]%
\\
\>[B]{}\hsindent{4}{}\<[4]%
\>[4]{}((\Varid{f}\;\Varid{y}\mathbin{:})\mathbin{\cdot}\Varid{map}\;\Varid{f})\mathrel{\raisebox{0.5\depth}{\scaleobj{0.5}{\langle}} \scaleobj{0.8}{\$} \raisebox{0.5\depth}{\scaleobj{0.5}{\rangle}}}\Varid{insert}\;\Varid{x}\;\Varid{xs}{}\<[E]%
\\
\>[B]{}\mathbin{=}{}\<[8]%
\>[8]{}\mbox{\commentbegin  by \eqref{eq:comp-ap-ap}  \commentend}{}\<[E]%
\\
\>[B]{}\hsindent{4}{}\<[4]%
\>[4]{}(\Varid{f}\;\Varid{y}\mathbin{:})\mathrel{\raisebox{0.5\depth}{\scaleobj{0.5}{\langle}} \scaleobj{0.8}{\$} \raisebox{0.5\depth}{\scaleobj{0.5}{\rangle}}}(\Varid{map}\;\Varid{f}\mathrel{\raisebox{0.5\depth}{\scaleobj{0.5}{\langle}} \scaleobj{0.8}{\$} \raisebox{0.5\depth}{\scaleobj{0.5}{\rangle}}}\Varid{insert}\;\Varid{x}\;\Varid{xs}){}\<[E]%
\\
\>[B]{}\mathbin{=}{}\<[8]%
\>[8]{}\mbox{\commentbegin  induction  \commentend}{}\<[E]%
\\
\>[B]{}\hsindent{4}{}\<[4]%
\>[4]{}(\Varid{f}\;\Varid{y}\mathbin{:})\mathrel{\raisebox{0.5\depth}{\scaleobj{0.5}{\langle}} \scaleobj{0.8}{\$} \raisebox{0.5\depth}{\scaleobj{0.5}{\rangle}}}\Varid{insert}\;(\Varid{f}\;\Varid{x})\;(\Varid{map}\;\Varid{f}\;\Varid{xs})~~.{}\<[E]%
\ColumnHook
\end{hscode}\resethooks
Thus we have
\begin{hscode}\SaveRestoreHook
\column{B}{@{}>{\hspre}l<{\hspost}@{}}%
\column{4}{@{}>{\hspre}l<{\hspost}@{}}%
\column{8}{@{}>{\hspre}l<{\hspost}@{}}%
\column{E}{@{}>{\hspre}l<{\hspost}@{}}%
\>[4]{}\Varid{map}\;\Varid{f}\mathrel{\raisebox{0.5\depth}{\scaleobj{0.5}{\langle}} \scaleobj{0.8}{\$} \raisebox{0.5\depth}{\scaleobj{0.5}{\rangle}}}\Varid{insert}\;\Varid{x}\;(\Varid{y}\mathbin{:}\Varid{xs}){}\<[E]%
\\
\>[B]{}\mathbin{=}{}\<[8]%
\>[8]{}\mbox{\commentbegin  calculation above  \commentend}{}\<[E]%
\\
\>[B]{}\hsindent{4}{}\<[4]%
\>[4]{}\Varid{return}\;(\Varid{f}\;\Varid{x}\mathbin{:}\Varid{f}\;\Varid{y}\mathbin{:}\Varid{map}\;\Varid{f}\;\Varid{xs})\mathbin{\talloblong}((\Varid{f}\;\Varid{y}\mathbin{:})\mathrel{\raisebox{0.5\depth}{\scaleobj{0.5}{\langle}} \scaleobj{0.8}{\$} \raisebox{0.5\depth}{\scaleobj{0.5}{\rangle}}}(\Varid{insert}\;(\Varid{f}\;\Varid{x})\;(\Varid{map}\;\Varid{f}\;\Varid{xs}))){}\<[E]%
\\
\>[B]{}\mathbin{=}{}\<[8]%
\>[8]{}\mbox{\commentbegin  definitions of \ensuremath{\Varid{insert}} and \ensuremath{\Varid{map}}  \commentend}{}\<[E]%
\\
\>[B]{}\hsindent{4}{}\<[4]%
\>[4]{}\Varid{insert}\;(\Varid{f}\;\Varid{x})\;(\Varid{map}\;\Varid{f}\;(\Varid{y}\mathbin{:}\Varid{xs}))~~.{}\<[E]%
\ColumnHook
\end{hscode}\resethooks
\end{proof}

One may have noticed that the style of proof is familiar: replace \ensuremath{\Varid{return}\;\Varid{x}} by \ensuremath{[\mskip1.5mu \Varid{x}\mskip1.5mu]} and \ensuremath{(\talloblong)} by \ensuremath{(\mathbin{+\!\!\!\!\!+})}, the proof is more-or-less what one would do for a list version of \ensuremath{\Varid{insert}}. This is exactly the point: the style of proofs we use to do for pure programs still works for monadic programs, as long as the monad satisfies the demanded laws, be it a list, a more advanced implementation of non-determinism, or a monad having other effects.

A similar property relating \ensuremath{\Varid{perm}} and \ensuremath{\Varid{filter}} can be formulated.
\begin{lemma} \label{lemma:perm-filter}
  \ensuremath{\Varid{perm}\mathbin{\cdot}\Varid{filter}\;\Varid{p}\mathrel{=}\Varid{filter}\;\Varid{p}\mathrel{\raisebox{0.5\depth}{\scaleobj{0.5}{\langle \bullet \rangle}}}\Varid{perm}}.
\end{lemma}
Its proof is routine and omitted. Finally, in a number of occasions it helps to know that \ensuremath{\Varid{xs}} is a result of \ensuremath{\Varid{perm}\;\Varid{xs}}. The proof is also routine and omitted.
\begin{lemma} \label{lemma:perm-id}
For all \ensuremath{\Varid{xs}} we have that \ensuremath{\Varid{perm}\;\Varid{xs}\mathrel{=}\Varid{return}\;\Varid{xs}\mathbin{\talloblong}\Varid{m}} for some \ensuremath{\Varid{m}}.
\end{lemma}

\section{Spark Aggregation}
\label{sec:spark}


Spark~\cite{Zaharia:12:Resilient} is a popular platform for scalable distributed data-parallel computation based on a flexible programming environment with high-level APIs, considered by many as the successor of MapReduce.
In a typical Spark program, data is partitioned and stored distributively on read-only {\em Resilient Distributed Datasets} (RDDs) ---
we can think of it as a list of lists, where each sub-list is potentially stored on a remote node.
On an RDD one can apply operations, called \emph{combinators}, such as \ensuremath{\Varid{map}}, \ensuremath{\Varid{reduce}}, and \ensuremath{\Varid{aggregate}}.
The \ensuremath{\Varid{aggregate}} combinator, for example, takes user-defined functions \ensuremath{(\otimes)} and \ensuremath{(\oplus)}: \ensuremath{(\otimes)} accumulates a sub-result for each data partition while \ensuremath{(\oplus)} merges sub-results across different partitions.

Programming in Spark, however, can be tricky.
Since sub-results are computed across partitions concurrently, the order of their applications varies on different executions.
Aggregation in Spark is therefore inherently non-deterministic.
An example from \citet{Lengal:17:Executable} showed that computing the integral of $x^{73}$ from $x=-2$ to $x=2$, which should be $0$, using a function in the Spark machine learning library, yields results ranging from $-8192.0$ to $12288.0$ in different runs.
It is thus desirable to find out conditions, which Spark's documentation does not specify formally, under which a Spark computation yields deterministic outcomes.

\subsection{List Homomorphism}

Since a Spark aggregation is typically used to computes a \emph{list homomorphism}~\cite{Bird:87:Introduction},
we digress a little in this section to give a brief review and present some results that we will use.
A function \ensuremath{\Varid{h}\mathbin{::}\Conid{List}\;\Varid{a}\to \Varid{b}} is called a list homomorphism if there exists \ensuremath{\Varid{z}\mathbin{::}\Varid{b}}, \ensuremath{\Varid{k}\mathbin{::}\Varid{a}\to \Varid{b}}, and \ensuremath{(\oplus)\mathbin{::}\Varid{b}\to \Varid{b}\to \Varid{b}} such that:
\begin{hscode}\SaveRestoreHook
\column{B}{@{}>{\hspre}l<{\hspost}@{}}%
\column{15}{@{}>{\hspre}l<{\hspost}@{}}%
\column{E}{@{}>{\hspre}l<{\hspost}@{}}%
\>[B]{}\Varid{h}\;[\mskip1.5mu \mskip1.5mu]{}\<[15]%
\>[15]{}\mathrel{=}\Varid{z}{}\<[E]%
\\
\>[B]{}\Varid{h}\;[\mskip1.5mu \Varid{x}\mskip1.5mu]{}\<[15]%
\>[15]{}\mathrel{=}\Varid{k}\;\Varid{x}{}\<[E]%
\\
\>[B]{}\Varid{h}\;(\Varid{xs}\mathbin{+\!\!\!\!\!+}\Varid{ys}){}\<[15]%
\>[15]{}\mathrel{=}\Varid{h}\;\Varid{xs}\mathbin{\oplus}\Varid{h}\;\Varid{ys}~~.{}\<[E]%
\ColumnHook
\end{hscode}\resethooks
That \ensuremath{\Varid{h}} is such a list homomorphism is denoted by \ensuremath{\Varid{h}\mathrel{=}\Varid{hom}\;(\oplus)\;\Varid{k}\;\Varid{z}}.
Note that the properties above implicitly demand that \ensuremath{(\oplus)} be associative with \ensuremath{\Varid{z}} as its identity element.

Lemma~\ref{lemma:hom-concat} and \ref{lemma:foldr-hom} below are about when a computation defined in terms of \ensuremath{\Varid{foldr}} is actually a list homomorphism. In Lemma~\ref{lemma:foldr-hom}, \ensuremath{\Varid{img}\;\Varid{f}} denotes the image of a function \ensuremath{\Varid{f}}.
\begin{lemma} \label{lemma:hom-concat}
\ensuremath{\Varid{h}\mathrel{=}\Varid{hom}\;(\oplus)\;(\Varid{h}\mathbin{\cdot}\Varid{wrap})\;\Varid{z}} if and only if \ensuremath{\Varid{foldr}\;(\oplus)\;\Varid{z}\mathbin{\cdot}\Varid{map}\;\Varid{h}\mathrel{=}\Varid{h}\mathbin{\cdot}\Varid{concat}}, where \ensuremath{\Varid{wrap}\;\Varid{x}\mathrel{=}[\mskip1.5mu \Varid{x}\mskip1.5mu]}.
\end{lemma}
\begin{lemma} \label{lemma:foldr-hom}
Let \ensuremath{(\oplus)\mathbin{::}\Varid{b}\to \Varid{b}\to \Varid{b}} be associative on \ensuremath{\Varid{img}\;(\Varid{foldr}\;(\otimes)\;\Varid{z})} with \ensuremath{\Varid{z}} as its identity, where \ensuremath{(\otimes)\mathbin{::}\Varid{a}\to \Varid{b}\to \Varid{b}}.
We have \ensuremath{\Varid{foldr}\;(\otimes)\;\Varid{z}\mathrel{=}\Varid{hom}\;(\oplus)\;(\mathbin{\otimes}\Varid{z})\;\Varid{z}} if and only if
\ensuremath{\Varid{x}\mathbin{\otimes}(\Varid{y}\mathbin{\oplus}\Varid{w})\mathrel{=}(\Varid{x}\mathbin{\otimes}\Varid{y})\mathbin{\oplus}\Varid{w}} for all
\ensuremath{\Varid{x}\mathbin{::}\Varid{a}} and \ensuremath{\Varid{y},\Varid{w}\in \Varid{img}\;(\Varid{foldr}\;(\otimes)\;\Varid{z})}.
\end{lemma}
Notice, in Lemma~\ref{lemma:foldr-hom}, that \ensuremath{(\mathbin{\otimes}\Varid{z})\mathrel{=}\Varid{foldr}\;(\otimes)\;\Varid{z}\mathbin{\cdot}\Varid{wrap}}.
Proofs of both lemmas are interesting exercises, albeit being a bit off-topic.
They are recorded in Appendix~\ref{sec:misc-proofs}.

\subsection{Formalisation and Results}

Distributed collections of data are represented by {\em Resilient Distributed Datasets} (RDDs) in Spark. Informally, an RDD is a collection of data entries; these data entries are further divided into partitions stored on different machines. Abstractly, an \ensuremath{\Conid{RDD}} can be seen as a list of lists:
\begin{hscode}\SaveRestoreHook
\column{B}{@{}>{\hspre}l<{\hspost}@{}}%
\column{19}{@{}>{\hspre}l<{\hspost}@{}}%
\column{E}{@{}>{\hspre}l<{\hspost}@{}}%
\>[B]{}\mathbf{type}\;\Conid{Partition}\;\Varid{a}{}\<[19]%
\>[19]{}\mathrel{=}[\mskip1.5mu \Varid{a}\mskip1.5mu]~~,{}\<[E]%
\\
\>[B]{}\mathbf{type}\;\Conid{RDD}\;\Varid{a}{}\<[19]%
\>[19]{}\mathrel{=}[\mskip1.5mu \Conid{Partition}\;\Varid{a}\mskip1.5mu]~~,{}\<[E]%
\ColumnHook
\end{hscode}\resethooks
where each \ensuremath{\Conid{Partition}} may be stored in a different machine.

While Spark provides a collection of \emph{combinators} (functions on \ensuremath{\Conid{RDD}}s that are designed to be composed to form larger programs), in this report we focus on a particular one, \ensuremath{\Varid{aggregate}}. It can be seen as a parallel implementation \ensuremath{\Varid{foldr}}. The combinator processes an \ensuremath{\Conid{RDD}} in two levels: each partition is first processed locally on one machine by \ensuremath{\Varid{foldr}\;(\otimes)\;\Varid{z}}. The sub-results are then communicated and combined ---
this second step can be think of as another \ensuremath{\Varid{foldr}} with \ensuremath{(\oplus)}.%
\footnote{In fact, the actual Spark aggregation (and that modelled in \citet{Lengal:17:Executable}) are like \ensuremath{\Varid{foldl}}.
For convenience in our proofs we see all list operations the other way round and use \ensuremath{\Varid{foldr}}. This is not a fundamental difference.}

Spark programmers like to assume that their programs are deterministic. To exploit concurrency, however, the sub-results from each machine might be processed in arbitrary order and the result could be non-deterministic.
The following is our characterisation of \ensuremath{\Varid{aggregate}}, where we use \ensuremath{\Varid{perm}} to model the fact that sub-results from each machine are processed in unknown order:
\begin{hscode}\SaveRestoreHook
\column{B}{@{}>{\hspre}l<{\hspost}@{}}%
\column{E}{@{}>{\hspre}l<{\hspost}@{}}%
\>[B]{}\Varid{aggregate}\mathbin{::}\Varid{b}\to (\Varid{a}\to \Varid{b}\to \Varid{b})\to (\Varid{b}\to \Varid{b}\to \Varid{b})\to \Conid{RDD}\;\Varid{a}\to \Conid{M}\;\Varid{b}{}\<[E]%
\\
\>[B]{}\Varid{aggregate}\;\Varid{z}\;(\otimes)\;(\oplus)\mathrel{=}\Varid{foldr}\;(\oplus)\;\Varid{z}\mathrel{\raisebox{0.5\depth}{\scaleobj{0.5}{\langle \bullet \rangle}}}(\Varid{perm}\mathbin{\cdot}\Varid{map}\;(\Varid{foldr}\;(\otimes)\;\Varid{z}))~~.{}\<[E]%
\ColumnHook
\end{hscode}\resethooks
It is clear from the types that \ensuremath{\Varid{foldr}\;(\otimes)\;\Varid{z}} and \ensuremath{\Varid{foldr}\;(\oplus)\;\Varid{z}} are pure computations, and non-determinism is introduced solely by \ensuremath{\Varid{perm}}.

\paragraph{Deterministic Aggregation}
We are interested in finding out conditions under which \ensuremath{\Varid{aggregate}} produces deterministic outcomes.
\begin{theorem}
\label{thm:aggregate-det}
Given \ensuremath{(\otimes)\mathbin{::}\Varid{a}\to \Varid{b}\to \Varid{b}} and \ensuremath{(\oplus)\mathbin{::}\Varid{b}\to \Varid{b}\to \Varid{b}}, where \ensuremath{(\oplus)} is associative and commutative, we have:\begin{hscode}\SaveRestoreHook
\column{B}{@{}>{\hspre}l<{\hspost}@{}}%
\column{3}{@{}>{\hspre}l<{\hspost}@{}}%
\column{E}{@{}>{\hspre}l<{\hspost}@{}}%
\>[3]{}\Varid{aggregate}\;\Varid{z}\;(\otimes)\;(\oplus)\mathrel{=}\Varid{return}\mathbin{\cdot}\Varid{foldr}\;(\oplus)\;\Varid{z}\mathbin{\cdot}\Varid{map}\;(\Varid{foldr}\;(\otimes)\;\Varid{z})~~.{}\<[E]%
\ColumnHook
\end{hscode}\resethooks
\end{theorem}
\begin{proof} We reason:
\begin{hscode}\SaveRestoreHook
\column{B}{@{}>{\hspre}l<{\hspost}@{}}%
\column{4}{@{}>{\hspre}l<{\hspost}@{}}%
\column{8}{@{}>{\hspre}l<{\hspost}@{}}%
\column{E}{@{}>{\hspre}l<{\hspost}@{}}%
\>[4]{}\Varid{aggregate}\;\Varid{z}\;(\otimes)\;(\oplus){}\<[E]%
\\
\>[B]{}\mathbin{=}{}\<[8]%
\>[8]{}\mbox{\commentbegin  definition of \ensuremath{\Varid{aggregate}}  \commentend}{}\<[E]%
\\
\>[B]{}\hsindent{4}{}\<[4]%
\>[4]{}\Varid{foldr}\;(\oplus)\;\Varid{z}\mathrel{\raisebox{0.5\depth}{\scaleobj{0.5}{\langle \bullet \rangle}}}(\Varid{perm}\mathbin{\cdot}\Varid{map}\;(\Varid{foldr}\;(\otimes)\;\Varid{z})){}\<[E]%
\\
\>[B]{}\mathbin{=}{}\<[8]%
\>[8]{}\mbox{\commentbegin  by \eqref{eq:mcomp-comp-mcomp}  \commentend}{}\<[E]%
\\
\>[B]{}\hsindent{4}{}\<[4]%
\>[4]{}(\Varid{foldr}\;(\oplus)\;\Varid{z}\mathrel{\raisebox{0.5\depth}{\scaleobj{0.5}{\langle \bullet \rangle}}}\Varid{perm})\mathbin{\cdot}\Varid{map}\;(\Varid{foldr}\;(\otimes)\;\Varid{z}){}\<[E]%
\\
\>[B]{}\mathbin{=}{}\<[8]%
\>[8]{}\mbox{\commentbegin   Lemma~\ref{lemma:fold-perm}, since \ensuremath{(\oplus)} is associative and commutative  \commentend}{}\<[E]%
\\
\>[B]{}\hsindent{4}{}\<[4]%
\>[4]{}\Varid{return}\mathbin{\cdot}\Varid{foldr}\;(\oplus)\;\Varid{z}\mathbin{\cdot}\Varid{map}\;(\Varid{foldr}\;(\otimes)\;\Varid{z})~~.{}\<[E]%
\ColumnHook
\end{hscode}\resethooks
\end{proof}

The following corollary summaries the results and present conditions under which \ensuremath{\Varid{aggregate}} computes a homomorphism.

\begin{corollary}
\label{cor:aggregate-det-hom}
\ensuremath{\Varid{aggregate}\;\Varid{z}\;(\otimes)\;(\oplus)\mathrel{=}\Varid{return}\mathbin{\cdot}\Varid{hom}\;(\oplus)\;(\mathbin{\otimes}\Varid{z})\;\Varid{z}\mathbin{\cdot}\Varid{concat}}, provided that
\ensuremath{(\oplus)} is associative, commutative, and has \ensuremath{\Varid{z}} as identity, and that \ensuremath{\Varid{x}\mathbin{\otimes}(\Varid{y}\mathbin{\oplus}\Varid{w})\mathrel{=}(\Varid{x}\mathbin{\otimes}\Varid{y})\mathbin{\oplus}\Varid{w}} for all
\ensuremath{\Varid{x}\mathbin{::}\Varid{a}} and \ensuremath{\Varid{y},\Varid{w}\in \Varid{img}\;(\Varid{foldr}\;(\otimes)\;\Varid{z})}.
\end{corollary}
\begin{proof} We reason:
\begin{hscode}\SaveRestoreHook
\column{B}{@{}>{\hspre}l<{\hspost}@{}}%
\column{4}{@{}>{\hspre}l<{\hspost}@{}}%
\column{8}{@{}>{\hspre}l<{\hspost}@{}}%
\column{E}{@{}>{\hspre}l<{\hspost}@{}}%
\>[4]{}\Varid{aggregate}\;\Varid{z}\;(\otimes)\;(\oplus){}\<[E]%
\\
\>[B]{}\mathbin{=}{}\<[8]%
\>[8]{}\mbox{\commentbegin  Theorem~\ref{thm:aggregate-det}  \commentend}{}\<[E]%
\\
\>[B]{}\hsindent{4}{}\<[4]%
\>[4]{}\Varid{return}\mathbin{\cdot}\Varid{foldr}\;(\oplus)\;\Varid{z}\mathbin{\cdot}\Varid{map}\;(\Varid{foldr}\;(\otimes)\;\Varid{z}){}\<[E]%
\\
\>[B]{}\mathbin{=}{}\<[8]%
\>[8]{}\mbox{\commentbegin  \ensuremath{\Varid{foldr}\;(\otimes)\;\Varid{z}\mathrel{=}\Varid{hom}\;(\oplus)\;(\mathbin{\otimes}\Varid{z})\;\Varid{z}} by Lemma~\ref{lemma:foldr-hom}; Lemma~\ref{lemma:hom-concat}  \commentend}{}\<[E]%
\\
\>[B]{}\hsindent{4}{}\<[4]%
\>[4]{}\Varid{return}\mathbin{\cdot}\Varid{hom}\;(\oplus)\;(\mathbin{\otimes}\Varid{z})\;\Varid{z}\mathbin{\cdot}\Varid{concat}~~.{}\<[E]%
\ColumnHook
\end{hscode}\resethooks
\end{proof}

\paragraph{Determinism Implies Homomorphism}
The final part of the report deals with an opposite question:
what can we infer if we know that \ensuremath{\Varid{aggregate}} is deterministic?
To answer that, however, we need to assume two more properties:
\begin{align}
 \ensuremath{\Varid{m}_{1}\mathbin{\talloblong}\Varid{m}_{2}\mathrel{=}\Varid{return}\;\Varid{x}} ~~&\Rightarrow~~ \ensuremath{\Varid{m}_{1}\mathrel{=}\Varid{m}_{2}\mathrel{=}\Varid{return}\;\Varid{x}} \mbox{.}
  \label{eq:mplus-return}\\
  \ensuremath{\Varid{return}\;\Varid{x}_{1}\mathrel{=}\Varid{return}\;\Varid{x}_{2}} ~~&\Rightarrow~~ \ensuremath{\Varid{x}_{1}\mathrel{=}\Varid{x}_{2}} \mbox{.}
  \label{eq:return-injective}
\end{align}
Property \eqref{eq:mplus-return} can be seen as the other direction of idempotency of \ensuremath{(\talloblong)},
while \eqref{eq:return-injective} states that \ensuremath{\Varid{return}} is injective.

The following lemma can be understood this way: when \ensuremath{\Varid{aggregate}\;\Varid{z}\;(\otimes)\;(\oplus)}, which could be non-deterministic, can be performed by a deterministic function, the operator \ensuremath{(\oplus)} should be insensitive to ordering:
\begin{lemma}
\label{lemma:aggregate-det-reasoning}
If \ensuremath{\Varid{aggregate}\;\Varid{z}\;(\otimes)\;(\oplus)\mathrel{=}\Varid{return}\mathbin{\cdot}\Varid{foldr}\;(\otimes)\;\Varid{z}\mathbin{\cdot}\Varid{concat}},
and \ensuremath{\Varid{perm}\;\Varid{xss}\mathrel{=}\Varid{return}\;\Varid{yss}\mathbin{\talloblong}\Varid{m}} for some \ensuremath{\Varid{m}}, we have
\begin{hscode}\SaveRestoreHook
\column{B}{@{}>{\hspre}l<{\hspost}@{}}%
\column{3}{@{}>{\hspre}l<{\hspost}@{}}%
\column{5}{@{}>{\hspre}l<{\hspost}@{}}%
\column{7}{@{}>{\hspre}l<{\hspost}@{}}%
\column{E}{@{}>{\hspre}l<{\hspost}@{}}%
\>[3]{}\Varid{foldr}\;(\otimes)\;\Varid{z}\;(\Varid{concat}\;\Varid{xss})\mathbin{=}{}\<[E]%
\\
\>[3]{}\hsindent{2}{}\<[5]%
\>[5]{}\Varid{foldr}\;(\oplus)\;\Varid{z}\;(\Varid{map}\;(\Varid{foldr}\;(\otimes)\;\Varid{z})\;\Varid{xss})\mathbin{=}{}\<[E]%
\\
\>[5]{}\hsindent{2}{}\<[7]%
\>[7]{}\Varid{foldr}\;(\oplus)\;\Varid{z}\;(\Varid{map}\;(\Varid{foldr}\;(\otimes)\;\Varid{z})\;\Varid{yss})~~.{}\<[E]%
\ColumnHook
\end{hscode}\resethooks
\end{lemma}
\begin{proof}
We reason:
\begin{hscode}\SaveRestoreHook
\column{B}{@{}>{\hspre}l<{\hspost}@{}}%
\column{4}{@{}>{\hspre}l<{\hspost}@{}}%
\column{7}{@{}>{\hspre}l<{\hspost}@{}}%
\column{9}{@{}>{\hspre}l<{\hspost}@{}}%
\column{E}{@{}>{\hspre}l<{\hspost}@{}}%
\>[4]{}\Varid{return}\mathbin{\cdot}\Varid{foldr}\;(\otimes)\;\Varid{z}\mathbin{\cdot}\Varid{concat}\mathbin{\$}\Varid{xss}{}\<[E]%
\\
\>[B]{}\mathbin{=}{}\<[9]%
\>[9]{}\mbox{\commentbegin  assumption  \commentend}{}\<[E]%
\\
\>[B]{}\hsindent{4}{}\<[4]%
\>[4]{}\Varid{aggregate}\;\Varid{z}\;(\otimes)\;(\oplus)\mathbin{\$}\Varid{xss}{}\<[E]%
\\
\>[B]{}\mathbin{=}{}\<[9]%
\>[9]{}\mbox{\commentbegin  definition of \ensuremath{\Varid{aggregate}}, Lemma~\ref{lemma:shuffle-map}, and \eqref{eq:comp-mcomp-mcomp}  \commentend}{}\<[E]%
\\
\>[B]{}\hsindent{4}{}\<[4]%
\>[4]{}(\Varid{foldr}\;(\oplus)\;\Varid{z}\mathbin{\cdot}\Varid{map}\;(\Varid{foldr}\;(\otimes)\;\Varid{z}))\mathrel{\raisebox{0.5\depth}{\scaleobj{0.5}{\langle}} \scaleobj{0.8}{\$} \raisebox{0.5\depth}{\scaleobj{0.5}{\rangle}}}\Varid{perm}\;\Varid{xss}{}\<[E]%
\\
\>[B]{}\mathbin{=}{}\<[9]%
\>[9]{}\mbox{\commentbegin  assumption: \ensuremath{\Varid{perm}\;\Varid{xss}\mathrel{=}\Varid{return}\;\Varid{yss}\mathbin{\talloblong}\Varid{m}}, by \eqref{eq:ap-mplus} and \eqref{eq:ap-return}  \commentend}{}\<[E]%
\\
\>[B]{}\hsindent{4}{}\<[4]%
\>[4]{}(\Varid{return}\mathbin{\cdot}\Varid{foldr}\;(\oplus)\;\Varid{z}\mathbin{\cdot}\Varid{map}\;(\Varid{foldr}\;(\otimes)\;\Varid{z})\mathbin{\$}\Varid{yss})\mathbin{\talloblong}{}\<[E]%
\\
\>[4]{}\hsindent{3}{}\<[7]%
\>[7]{}((\Varid{foldr}\;(\oplus)\;\Varid{z}\mathbin{\cdot}\Varid{map}\;(\Varid{foldr}\;(\otimes)\;\Varid{z}))\mathrel{\raisebox{0.5\depth}{\scaleobj{0.5}{\langle}} \scaleobj{0.8}{\$} \raisebox{0.5\depth}{\scaleobj{0.5}{\rangle}}}\Varid{m})~~.{}\<[E]%
\ColumnHook
\end{hscode}\resethooks
Thus by \eqref{eq:mplus-return} and \eqref{eq:return-injective},
\ensuremath{\Varid{foldr}\;(\otimes)\;\Varid{z}\mathbin{\cdot}\Varid{concat}\mathbin{\$}\Varid{xss}} equals
\ensuremath{\Varid{foldr}\;(\oplus)\;\Varid{z}\mathbin{\cdot}\Varid{map}\;(\Varid{foldr}\;(\otimes)\;\Varid{z})\mathbin{\$}\Varid{yss}}.
The former also equals
\ensuremath{\Varid{foldr}\;(\oplus)\;\Varid{z}\mathbin{\cdot}\Varid{map}\;(\Varid{foldr}\;(\otimes)\;\Varid{z})\mathbin{\$}\Varid{xss}} because,
by Lemma~\ref{lemma:perm-id},
\ensuremath{\Varid{perm}\;\Varid{xss}\mathrel{=}\Varid{return}\;\Varid{xss}\mathbin{\talloblong}\Varid{m}} for some \ensuremath{\Varid{m}}.
\end{proof}

Based on Lemma~\ref{lemma:aggregate-det-reasoning}, the following theorem explicitly states that \ensuremath{(\oplus)} should be associative, commutative, and has \ensuremath{\Varid{z}} as its identity in restricted domain.
\begin{theorem}
\label{thm:aggregate-cmonoid}
If \ensuremath{\Varid{aggregate}\;\Varid{z}\;(\otimes)\;(\oplus)\mathrel{=}\Varid{return}\mathbin{\cdot}\Varid{foldr}\;(\otimes)\;\Varid{z}\mathbin{\cdot}\Varid{concat}},
we have that \ensuremath{(\oplus)}, when restricted to values in \ensuremath{\Varid{img}\;(\Varid{foldr}\;(\otimes)\;\Varid{z})}, is associative, commutative, and has \ensuremath{\Varid{z}} as its identity.
\end{theorem}
\begin{proof}
In the discussion below, let \ensuremath{\Varid{x}},\ensuremath{\Varid{y}}, and \ensuremath{\Varid{w}} be in \ensuremath{\Varid{img}\;(\Varid{foldr}\;(\otimes)\;\Varid{z})}. That is, there exists \ensuremath{\Varid{xs}}, \ensuremath{\Varid{ys}}, and \ensuremath{\Varid{ws}} such that \ensuremath{\Varid{x}\mathrel{=}\Varid{foldr}\;(\otimes)\;\Varid{z}\;\Varid{xs}},
\ensuremath{\Varid{y}\mathrel{=}\Varid{foldr}\;(\otimes)\;\Varid{z}\;\Varid{ys}}, and \ensuremath{\Varid{w}\mathrel{=}\Varid{foldr}\;(\otimes)\;\Varid{z}\;\Varid{ws}}.

{\sc Identity}. We reason:
\begin{hscode}\SaveRestoreHook
\column{B}{@{}>{\hspre}l<{\hspost}@{}}%
\column{7}{@{}>{\hspre}l<{\hspost}@{}}%
\column{9}{@{}>{\hspre}l<{\hspost}@{}}%
\column{E}{@{}>{\hspre}l<{\hspost}@{}}%
\>[7]{}\Varid{y}{}\<[E]%
\\
\>[B]{}\mathbin{=}{}\<[7]%
\>[7]{}\Varid{foldr}\;(\otimes)\;\Varid{z}\;(\Varid{concat}\;[\mskip1.5mu \Varid{xs}\mskip1.5mu]){}\<[E]%
\\
\>[B]{}\mathbin{=}{}\<[9]%
\>[9]{}\mbox{\commentbegin  \ensuremath{\Varid{perm}\;[\mskip1.5mu \Varid{xs}\mskip1.5mu]\mathrel{=}\Varid{return}\;[\mskip1.5mu \Varid{xs}\mskip1.5mu]\mathbin{\talloblong}\emptyset}, Lemma \ref{lemma:aggregate-det-reasoning}  \commentend}{}\<[E]%
\\
\>[B]{}\hsindent{7}{}\<[7]%
\>[7]{}\Varid{foldr}\;(\oplus)\;\Varid{z}\;(\Varid{map}\;(\Varid{foldr}\;(\otimes)\;\Varid{z})\;[\mskip1.5mu \Varid{xs}\mskip1.5mu]){}\<[E]%
\\
\>[B]{}\mathbin{=}{}\<[7]%
\>[7]{}\Varid{y}\mathbin{\oplus}\Varid{z}~~.{}\<[E]%
\ColumnHook
\end{hscode}\resethooks
Thus \ensuremath{\Varid{z}} is a right identity of \ensuremath{(\oplus)}. Similarly,
\begin{hscode}\SaveRestoreHook
\column{B}{@{}>{\hspre}l<{\hspost}@{}}%
\column{4}{@{}>{\hspre}l<{\hspost}@{}}%
\column{7}{@{}>{\hspre}l<{\hspost}@{}}%
\column{8}{@{}>{\hspre}l<{\hspost}@{}}%
\column{9}{@{}>{\hspre}l<{\hspost}@{}}%
\column{E}{@{}>{\hspre}l<{\hspost}@{}}%
\>[7]{}\Varid{y}{}\<[E]%
\\
\>[B]{}\mathbin{=}{}\<[7]%
\>[7]{}\Varid{foldr}\;(\otimes)\;\Varid{z}\;(\Varid{concat}\;[\mskip1.5mu [\mskip1.5mu \mskip1.5mu],\Varid{xs}\mskip1.5mu]){}\<[E]%
\\
\>[B]{}\mathbin{=}{}\<[9]%
\>[9]{}\mbox{\commentbegin  \ensuremath{\Varid{perm}\;[\mskip1.5mu [\mskip1.5mu \mskip1.5mu],\Varid{xs}\mskip1.5mu]\mathrel{=}\Varid{return}\;[\mskip1.5mu [\mskip1.5mu \mskip1.5mu],\Varid{xs}\mskip1.5mu]\mathbin{\talloblong}\Varid{m}}, Lemma \ref{lemma:aggregate-det-reasoning}  \commentend}{}\<[E]%
\\
\>[B]{}\hsindent{4}{}\<[4]%
\>[4]{}\Varid{foldr}\;(\oplus)\;\Varid{z}\;(\Varid{map}\;(\Varid{foldr}\;(\otimes)\;\Varid{z})\;[\mskip1.5mu [\mskip1.5mu \mskip1.5mu],\Varid{xs}\mskip1.5mu]){}\<[E]%
\\
\>[B]{}\mathbin{=}{}\<[7]%
\>[7]{}\Varid{z}\mathbin{\oplus}(\Varid{y}\mathbin{\oplus}\Varid{z}){}\<[E]%
\\
\>[B]{}\mathbin{=}{}\<[8]%
\>[8]{}\mbox{\commentbegin  \ensuremath{\Varid{z}} is a right identity of \ensuremath{(\oplus)}  \commentend}{}\<[E]%
\\
\>[B]{}\hsindent{7}{}\<[7]%
\>[7]{}\Varid{z}\mathbin{\oplus}\Varid{y}~~.{}\<[E]%
\ColumnHook
\end{hscode}\resethooks
Thus \ensuremath{\Varid{z}} is also a left identity of \ensuremath{(\oplus)}.

{\sc Commutativity}. We reason:
\begin{hscode}\SaveRestoreHook
\column{B}{@{}>{\hspre}l<{\hspost}@{}}%
\column{7}{@{}>{\hspre}l<{\hspost}@{}}%
\column{9}{@{}>{\hspre}l<{\hspost}@{}}%
\column{E}{@{}>{\hspre}l<{\hspost}@{}}%
\>[7]{}\Varid{x}\mathbin{\oplus}\Varid{y}{}\<[E]%
\\
\>[B]{}\mathbin{=}{}\<[9]%
\>[9]{}\mbox{\commentbegin  \ensuremath{\Varid{z}} is a right identity  \commentend}{}\<[E]%
\\
\>[B]{}\hsindent{7}{}\<[7]%
\>[7]{}\Varid{x}\mathbin{\oplus}(\Varid{y}\mathbin{\oplus}\Varid{z}){}\<[E]%
\\
\>[B]{}\mathbin{=}{}\<[7]%
\>[7]{}\Varid{foldr}\;(\oplus)\;\Varid{z}\;(\Varid{map}\;(\Varid{foldr}\;(\otimes)\;\Varid{z})\;[\mskip1.5mu \Varid{xs},\Varid{ys}\mskip1.5mu]){}\<[E]%
\\
\>[B]{}\mathbin{=}{}\<[9]%
\>[9]{}\mbox{\commentbegin  \ensuremath{\Varid{perm}\;[\mskip1.5mu \Varid{xs},\Varid{ys}\mskip1.5mu]\mathrel{=}\Varid{return}\;[\mskip1.5mu \Varid{ys},\Varid{xs}\mskip1.5mu]\mathbin{\talloblong}\Varid{m}}, Lemma \ref{lemma:aggregate-det-reasoning}  \commentend}{}\<[E]%
\\
\>[B]{}\hsindent{7}{}\<[7]%
\>[7]{}\Varid{foldr}\;(\oplus)\;\Varid{z}\;(\Varid{map}\;(\Varid{foldr}\;(\otimes)\;\Varid{z})\;[\mskip1.5mu \Varid{ys},\Varid{xs}\mskip1.5mu]){}\<[E]%
\\
\>[B]{}\mathbin{=}{}\<[7]%
\>[7]{}\Varid{y}\mathbin{\oplus}(\Varid{x}\mathbin{\oplus}\Varid{z}){}\<[E]%
\\
\>[B]{}\mathbin{=}{}\<[9]%
\>[9]{}\mbox{\commentbegin  \ensuremath{\Varid{z}} is a right identity  \commentend}{}\<[E]%
\\
\>[B]{}\hsindent{7}{}\<[7]%
\>[7]{}\Varid{y}\mathbin{\oplus}\Varid{x}~~.{}\<[E]%
\ColumnHook
\end{hscode}\resethooks

{\sc Associativity}. We reason:
\begin{hscode}\SaveRestoreHook
\column{B}{@{}>{\hspre}l<{\hspost}@{}}%
\column{7}{@{}>{\hspre}l<{\hspost}@{}}%
\column{9}{@{}>{\hspre}l<{\hspost}@{}}%
\column{E}{@{}>{\hspre}l<{\hspost}@{}}%
\>[7]{}\Varid{x}\mathbin{\oplus}(\Varid{y}\mathbin{\oplus}\Varid{w}){}\<[E]%
\\
\>[B]{}\mathbin{=}{}\<[9]%
\>[9]{}\mbox{\commentbegin  \ensuremath{\Varid{z}} is a right identity  \commentend}{}\<[E]%
\\
\>[B]{}\hsindent{7}{}\<[7]%
\>[7]{}\Varid{x}\mathbin{\oplus}(\Varid{y}\mathbin{\oplus}(\Varid{w}\mathbin{\oplus}\Varid{z})){}\<[E]%
\\
\>[B]{}\mathbin{=}{}\<[7]%
\>[7]{}\Varid{foldr}\;(\oplus)\;\Varid{z}\;(\Varid{map}\;(\Varid{foldr}\;(\otimes)\;\Varid{z})\;[\mskip1.5mu \Varid{xs},\Varid{ys},\Varid{ws}\mskip1.5mu]){}\<[E]%
\\
\>[B]{}\mathbin{=}{}\<[9]%
\>[9]{}\mbox{\commentbegin  \ensuremath{(\oplus)} commutative  \commentend}{}\<[E]%
\\
\>[B]{}\hsindent{7}{}\<[7]%
\>[7]{}\Varid{foldr}\;(\oplus)\;\Varid{z}\;(\Varid{map}\;(\Varid{foldr}\;(\otimes)\;\Varid{z})\;[\mskip1.5mu \Varid{ws},\Varid{xs},\Varid{ys}\mskip1.5mu]){}\<[E]%
\\
\>[B]{}\mathbin{=}{}\<[7]%
\>[7]{}\Varid{w}\mathbin{\oplus}(\Varid{x}\mathbin{\oplus}(\Varid{y}\mathbin{\oplus}\Varid{z})){}\<[E]%
\\
\>[B]{}\mathbin{=}{}\<[9]%
\>[9]{}\mbox{\commentbegin  \ensuremath{\Varid{z}} is a right identity  \commentend}{}\<[E]%
\\
\>[B]{}\hsindent{7}{}\<[7]%
\>[7]{}\Varid{w}\mathbin{\oplus}(\Varid{x}\mathbin{\oplus}\Varid{y}){}\<[E]%
\\
\>[B]{}\mathbin{=}{}\<[9]%
\>[9]{}\mbox{\commentbegin  \ensuremath{(\oplus)} commutative  \commentend}{}\<[E]%
\\
\>[B]{}\hsindent{7}{}\<[7]%
\>[7]{}(\Varid{x}\mathbin{\oplus}\Varid{y})\mathbin{\oplus}\Varid{w}~~.{}\<[E]%
\ColumnHook
\end{hscode}\resethooks
\end{proof}

\begin{theorem}
\label{thm:aggregate-hom}
If \ensuremath{\Varid{aggregate}\;\Varid{z}\;(\otimes)\;(\oplus)\mathrel{=}\Varid{return}\mathbin{\cdot}\Varid{foldr}\;(\otimes)\;\Varid{z}\mathbin{\cdot}\Varid{concat}},
we have \ensuremath{\Varid{foldr}\;(\otimes)\;\Varid{z}\mathrel{=}\Varid{hom}\;(\oplus)\;(\mathbin{\otimes}\Varid{z})\;\Varid{z}}.
\end{theorem}
\begin{proof}
Apparently \ensuremath{\Varid{foldr}\;(\otimes)\;\Varid{z}\;[\mskip1.5mu \mskip1.5mu]\mathrel{=}\Varid{z}} and
\ensuremath{\Varid{foldr}\;(\otimes)\;\Varid{z}\;[\mskip1.5mu \Varid{x}\mskip1.5mu]\mathrel{=}\Varid{x}\mathbin{\otimes}\Varid{z}}.
We are left with proving the case for concatenation.
\begin{hscode}\SaveRestoreHook
\column{B}{@{}>{\hspre}l<{\hspost}@{}}%
\column{7}{@{}>{\hspre}l<{\hspost}@{}}%
\column{9}{@{}>{\hspre}l<{\hspost}@{}}%
\column{E}{@{}>{\hspre}l<{\hspost}@{}}%
\>[7]{}\Varid{foldr}\;(\otimes)\;\Varid{z}\;(\Varid{xs}\mathbin{+\!\!\!\!\!+}\Varid{ys}){}\<[E]%
\\
\>[B]{}\mathbin{=}{}\<[7]%
\>[7]{}\Varid{foldr}\;(\otimes)\;\Varid{z}\;(\Varid{concat}\;[\mskip1.5mu \Varid{xs},\Varid{ys}\mskip1.5mu]){}\<[E]%
\\
\>[B]{}\mathbin{=}{}\<[9]%
\>[9]{}\mbox{\commentbegin  Lemma \ref{lemma:aggregate-det-reasoning}  \commentend}{}\<[E]%
\\
\>[B]{}\hsindent{7}{}\<[7]%
\>[7]{}\Varid{foldr}\;(\oplus)\;\Varid{z}\;(\Varid{map}\;(\Varid{foldr}\;(\otimes)\;\Varid{z})\;[\mskip1.5mu \Varid{xs},\Varid{ys}\mskip1.5mu]){}\<[E]%
\\
\>[B]{}\mathbin{=}{}\<[7]%
\>[7]{}\Varid{foldr}\;(\otimes)\;\Varid{z}\;\Varid{xs}\mathbin{\oplus}(\Varid{foldr}\;(\otimes)\;\Varid{z}\;\Varid{ys}\mathbin{\oplus}\Varid{z}){}\<[E]%
\\
\>[B]{}\mathbin{=}{}\<[9]%
\>[9]{}\mbox{\commentbegin  Theorem \ref{thm:aggregate-cmonoid}, \ensuremath{\Varid{z}} is identity  \commentend}{}\<[E]%
\\
\>[B]{}\hsindent{7}{}\<[7]%
\>[7]{}\Varid{foldr}\;(\otimes)\;\Varid{z}\;\Varid{xs}\mathbin{\oplus}\Varid{foldr}\;(\otimes)\;\Varid{z}\;\Varid{ys}~~.{}\<[E]%
\ColumnHook
\end{hscode}\resethooks
\end{proof}

\begin{corollary} Given \ensuremath{(\otimes)\mathbin{::}\Varid{a}\to \Varid{b}\to \Varid{b}} and \ensuremath{(\oplus)\mathbin{::}\Varid{b}\to \Varid{b}\to \Varid{b}}.
\ensuremath{\Varid{aggregate}\;\Varid{z}\;(\otimes)\;(\oplus)\mathrel{=}\Varid{return}\mathbin{\cdot}\Varid{foldr}\;(\otimes)\;\Varid{z}\mathbin{\cdot}\Varid{concat}} if and only
if \ensuremath{(\Varid{img}\;(\Varid{foldr}\;(\otimes)\;\Varid{z}),(\oplus),\Varid{z})} forms a commutative monoid, and
that \ensuremath{\Varid{foldr}\;(\otimes)\;\Varid{z}\mathrel{=}\Varid{hom}\;(\oplus)\;(\mathbin{\otimes}\Varid{z})\;\Varid{z}}.
\end{corollary}
\begin{proof}
A conclusion following from Theorem~\ref{thm:aggregate-det},
Theorem \ref{thm:aggregate-cmonoid}, and Theorem \ref{thm:aggregate-hom}.
\end{proof}

\subsubsection*{Acknowledgements}
In around late 2016,
Yu-Fang Chen, Chih-Duo Hong, Ond{\v{r}}ej Leng{\'{a}}l,
Nishant Sinha and Bow-Yaw Wang invited me into their project
formalising Spark. It was what inspired my interests in reasoning about monads, which led to a number of subsequent work.
The initial proofs of properties of \ensuremath{\Varid{aggregate}} and other combinators were done by Ond{\v{r}}ej Leng{\'{a}}l, without using monads.
\citet{Affeldt:19:Hierarchy} modelled a hierarchy of monadic effects in Coq. The formalisation was applied to verify a number of equational proofs of monadic programs, including some of the proofs in an earlier version of this report. I am solely responsible for any remaining errors, however.


\appendix

\section{Miscellaneous Proofs}
\label{sec:misc-proofs}

\paragraph{Proving \eqref{eq:comp-bind-ap}} \ensuremath{\Varid{f}\mathbin{\hstretch{0.7}{=\!\!<\!\!<}}(\Varid{g}\mathrel{\raisebox{0.5\depth}{\scaleobj{0.5}{\langle}} \scaleobj{0.8}{\$} \raisebox{0.5\depth}{\scaleobj{0.5}{\rangle}}}\Varid{m})\mathrel{=}(\Varid{f}\mathbin{\cdot}\Varid{g})\mathbin{\hstretch{0.7}{=\!\!<\!\!<}}\Varid{m}}.
\begin{proof} We reason:
\begin{hscode}\SaveRestoreHook
\column{B}{@{}>{\hspre}l<{\hspost}@{}}%
\column{7}{@{}>{\hspre}l<{\hspost}@{}}%
\column{9}{@{}>{\hspre}l<{\hspost}@{}}%
\column{E}{@{}>{\hspre}l<{\hspost}@{}}%
\>[7]{}\Varid{f}\mathbin{\hstretch{0.7}{=\!\!<\!\!<}}(\Varid{g}\mathrel{\raisebox{0.5\depth}{\scaleobj{0.5}{\langle}} \scaleobj{0.8}{\$} \raisebox{0.5\depth}{\scaleobj{0.5}{\rangle}}}\Varid{m}){}\<[E]%
\\
\>[B]{}\mathbin{=}{}\<[9]%
\>[9]{}\mbox{\commentbegin  definition of \ensuremath{(\mathrel{\raisebox{0.5\depth}{\scaleobj{0.5}{\langle}} \scaleobj{0.8}{\$} \raisebox{0.5\depth}{\scaleobj{0.5}{\rangle}}})}  \commentend}{}\<[E]%
\\
\>[B]{}\hsindent{7}{}\<[7]%
\>[7]{}\Varid{f}\mathbin{\hstretch{0.7}{=\!\!<\!\!<}}((\Varid{return}\mathbin{\cdot}\Varid{g})\mathbin{\hstretch{0.7}{=\!\!<\!\!<}}\Varid{m}){}\<[E]%
\\
\>[B]{}\mathbin{=}{}\<[9]%
\>[9]{}\mbox{\commentbegin  monad law \eqref{eq:monad-assoc}  \commentend}{}\<[E]%
\\
\>[B]{}\hsindent{7}{}\<[7]%
\>[7]{}(\lambda \Varid{x}\to \Varid{f}\mathbin{\hstretch{0.7}{=\!\!<\!\!<}}\Varid{return}\;(\Varid{g}\;\Varid{x}))\mathbin{\hstretch{0.7}{=\!\!<\!\!<}}\Varid{m}{}\<[E]%
\\
\>[B]{}\mathbin{=}{}\<[9]%
\>[9]{}\mbox{\commentbegin  monad law \eqref{eq:monad-bind-ret}  \commentend}{}\<[E]%
\\
\>[B]{}\hsindent{7}{}\<[7]%
\>[7]{}(\lambda \Varid{x}\to \Varid{f}\;(\Varid{g}\;\Varid{x}))\mathbin{\hstretch{0.7}{=\!\!<\!\!<}}\Varid{m}{}\<[E]%
\\
\>[B]{}\mathbin{=}{}\<[7]%
\>[7]{}(\Varid{f}\mathbin{\cdot}\Varid{g})\mathbin{\hstretch{0.7}{=\!\!<\!\!<}}\Varid{m}~~.{}\<[E]%
\ColumnHook
\end{hscode}\resethooks
\end{proof}

\paragraph{Proving \eqref{eq:comp-ap-ap}} \ensuremath{\Varid{f}\mathrel{\raisebox{0.5\depth}{\scaleobj{0.5}{\langle}} \scaleobj{0.8}{\$} \raisebox{0.5\depth}{\scaleobj{0.5}{\rangle}}}(\Varid{g}\mathrel{\raisebox{0.5\depth}{\scaleobj{0.5}{\langle}} \scaleobj{0.8}{\$} \raisebox{0.5\depth}{\scaleobj{0.5}{\rangle}}}\Varid{m})\mathrel{=}(\Varid{f}\mathbin{\cdot}\Varid{g})\mathrel{\raisebox{0.5\depth}{\scaleobj{0.5}{\langle}} \scaleobj{0.8}{\$} \raisebox{0.5\depth}{\scaleobj{0.5}{\rangle}}}\Varid{m}}.
\begin{proof} We reason:
\begin{hscode}\SaveRestoreHook
\column{B}{@{}>{\hspre}l<{\hspost}@{}}%
\column{4}{@{}>{\hspre}l<{\hspost}@{}}%
\column{9}{@{}>{\hspre}l<{\hspost}@{}}%
\column{E}{@{}>{\hspre}l<{\hspost}@{}}%
\>[4]{}\Varid{f}\mathrel{\raisebox{0.5\depth}{\scaleobj{0.5}{\langle}} \scaleobj{0.8}{\$} \raisebox{0.5\depth}{\scaleobj{0.5}{\rangle}}}(\Varid{g}\mathrel{\raisebox{0.5\depth}{\scaleobj{0.5}{\langle}} \scaleobj{0.8}{\$} \raisebox{0.5\depth}{\scaleobj{0.5}{\rangle}}}\Varid{m}){}\<[E]%
\\
\>[B]{}\mathbin{=}{}\<[9]%
\>[9]{}\mbox{\commentbegin  definition of \ensuremath{(\mathrel{\raisebox{0.5\depth}{\scaleobj{0.5}{\langle}} \scaleobj{0.8}{\$} \raisebox{0.5\depth}{\scaleobj{0.5}{\rangle}}})}  \commentend}{}\<[E]%
\\
\>[B]{}\hsindent{4}{}\<[4]%
\>[4]{}(\Varid{return}\mathbin{\cdot}\Varid{f})\mathbin{\hstretch{0.7}{=\!\!<\!\!<}}(\Varid{g}\mathrel{\raisebox{0.5\depth}{\scaleobj{0.5}{\langle}} \scaleobj{0.8}{\$} \raisebox{0.5\depth}{\scaleobj{0.5}{\rangle}}}\Varid{m}){}\<[E]%
\\
\>[B]{}\mathbin{=}{}\<[9]%
\>[9]{}\mbox{\commentbegin  by \eqref{eq:comp-bind-ap}  \commentend}{}\<[E]%
\\
\>[B]{}\hsindent{4}{}\<[4]%
\>[4]{}(\Varid{return}\mathbin{\cdot}\Varid{f}\mathbin{\cdot}\Varid{g})\mathbin{\hstretch{0.7}{=\!\!<\!\!<}}\Varid{m}{}\<[E]%
\\
\>[B]{}\mathbin{=}{}\<[9]%
\>[9]{}\mbox{\commentbegin  definition of \ensuremath{(\mathrel{\raisebox{0.5\depth}{\scaleobj{0.5}{\langle}} \scaleobj{0.8}{\$} \raisebox{0.5\depth}{\scaleobj{0.5}{\rangle}}})}  \commentend}{}\<[E]%
\\
\>[B]{}\hsindent{4}{}\<[4]%
\>[4]{}(\Varid{f}\mathbin{\cdot}\Varid{g})\mathrel{\raisebox{0.5\depth}{\scaleobj{0.5}{\langle}} \scaleobj{0.8}{\$} \raisebox{0.5\depth}{\scaleobj{0.5}{\rangle}}}\Varid{m}~~.{}\<[E]%
\ColumnHook
\end{hscode}\resethooks
\end{proof}

For the next results we prove a lemma:
\begin{align}
\ensuremath{(\Varid{f}\mathbin{\hstretch{0.7}{=\!\!<\!\!<}})\mathbin{\cdot}(\Varid{g}\mathbin{\hstretch{0.7}{=\!\!<\!\!<}})} &= \ensuremath{(((\Varid{f}\mathbin{\hstretch{0.7}{=\!\!<\!\!<}})\mathbin{\cdot}\Varid{g})\mathbin{\hstretch{0.7}{=\!\!<\!\!<}})} \mbox{~~.} \label{eq:bind-comp-bind}
\end{align}
\begin{hscode}\SaveRestoreHook
\column{B}{@{}>{\hspre}l<{\hspost}@{}}%
\column{4}{@{}>{\hspre}l<{\hspost}@{}}%
\column{9}{@{}>{\hspre}l<{\hspost}@{}}%
\column{E}{@{}>{\hspre}l<{\hspost}@{}}%
\>[4]{}(\Varid{f}\mathbin{\hstretch{0.7}{=\!\!<\!\!<}})\mathbin{\cdot}(\Varid{g}\mathbin{\hstretch{0.7}{=\!\!<\!\!<}}){}\<[E]%
\\
\>[B]{}\mathbin{=}{}\<[9]%
\>[9]{}\mbox{\commentbegin  $\eta$ intro.  \commentend}{}\<[E]%
\\
\>[B]{}\hsindent{4}{}\<[4]%
\>[4]{}(\lambda \Varid{m}\to \Varid{f}\mathbin{\hstretch{0.7}{=\!\!<\!\!<}}(\Varid{g}\mathbin{\hstretch{0.7}{=\!\!<\!\!<}}\Varid{m})){}\<[E]%
\\
\>[B]{}\mathbin{=}{}\<[9]%
\>[9]{}\mbox{\commentbegin  monad law \eqref{eq:monad-assoc}  \commentend}{}\<[E]%
\\
\>[B]{}\hsindent{4}{}\<[4]%
\>[4]{}(\lambda \Varid{m}\to (\lambda \Varid{y}\to \Varid{f}\mathbin{\hstretch{0.7}{=\!\!<\!\!<}}\Varid{g}\;\Varid{y})\mathbin{\hstretch{0.7}{=\!\!<\!\!<}}\Varid{m}){}\<[E]%
\\
\>[B]{}\mathbin{=}{}\<[9]%
\>[9]{}\mbox{\commentbegin  $\eta$ reduction  \commentend}{}\<[E]%
\\
\>[B]{}\hsindent{4}{}\<[4]%
\>[4]{}(((\Varid{f}\mathbin{\hstretch{0.7}{=\!\!<\!\!<}})\mathbin{\cdot}\Varid{g})\mathbin{\hstretch{0.7}{=\!\!<\!\!<}})~~.{}\<[E]%
\ColumnHook
\end{hscode}\resethooks


\paragraph{Proving \eqref{eq:comp-mcomp-mcomp}} \ensuremath{\Varid{f}\mathrel{\raisebox{0.5\depth}{\scaleobj{0.5}{\langle \bullet \rangle}}}(\Varid{g}\mathrel{\raisebox{0.5\depth}{\scaleobj{0.5}{\langle \bullet \rangle}}}\Varid{m})\mathrel{=}(\Varid{f}\mathbin{\cdot}\Varid{g})\mathrel{\raisebox{0.5\depth}{\scaleobj{0.5}{\langle \bullet \rangle}}}\Varid{m}}.
\begin{proof} We reason:
\begin{hscode}\SaveRestoreHook
\column{B}{@{}>{\hspre}l<{\hspost}@{}}%
\column{6}{@{}>{\hspre}l<{\hspost}@{}}%
\column{9}{@{}>{\hspre}l<{\hspost}@{}}%
\column{E}{@{}>{\hspre}l<{\hspost}@{}}%
\>[6]{}\Varid{f}\mathrel{\raisebox{0.5\depth}{\scaleobj{0.5}{\langle \bullet \rangle}}}(\Varid{g}\mathrel{\raisebox{0.5\depth}{\scaleobj{0.5}{\langle \bullet \rangle}}}\Varid{m}){}\<[E]%
\\
\>[B]{}\mathbin{=}{}\<[9]%
\>[9]{}\mbox{\commentbegin  definition of \ensuremath{(\mathrel{\raisebox{0.5\depth}{\scaleobj{0.5}{\langle \bullet \rangle}}})}  \commentend}{}\<[E]%
\\
\>[B]{}\hsindent{6}{}\<[6]%
\>[6]{}((\Varid{return}\mathbin{\cdot}\Varid{f})\mathbin{\hstretch{0.7}{=\!\!<\!\!<}})\mathbin{\cdot}((\Varid{return}\mathbin{\cdot}\Varid{g})\mathbin{\hstretch{0.7}{=\!\!<\!\!<}})\mathbin{\cdot}\Varid{m}{}\<[E]%
\\
\>[B]{}\mathbin{=}{}\<[9]%
\>[9]{}\mbox{\commentbegin  by \eqref{eq:bind-comp-bind}  \commentend}{}\<[E]%
\\
\>[B]{}\hsindent{6}{}\<[6]%
\>[6]{}((((\Varid{return}\mathbin{\cdot}\Varid{f})\mathbin{\hstretch{0.7}{=\!\!<\!\!<}})\mathbin{\cdot}\Varid{return}\mathbin{\cdot}\Varid{g})\mathbin{\hstretch{0.7}{=\!\!<\!\!<}})\mathbin{\cdot}\Varid{m}{}\<[E]%
\\
\>[B]{}\mathbin{=}{}\<[9]%
\>[9]{}\mbox{\commentbegin  monad law \eqref{eq:monad-bind-ret}  \commentend}{}\<[E]%
\\
\>[B]{}\hsindent{6}{}\<[6]%
\>[6]{}((\Varid{return}\mathbin{\cdot}\Varid{f}\mathbin{\cdot}\Varid{g})\mathbin{\hstretch{0.7}{=\!\!<\!\!<}})\mathbin{\cdot}\Varid{m}{}\<[E]%
\\
\>[B]{}\mathbin{=}{}\<[9]%
\>[9]{}\mbox{\commentbegin  definition of \ensuremath{(\mathrel{\raisebox{0.5\depth}{\scaleobj{0.5}{\langle \bullet \rangle}}})}  \commentend}{}\<[E]%
\\
\>[B]{}\hsindent{6}{}\<[6]%
\>[6]{}(\Varid{f}\mathbin{\cdot}\Varid{g})\mathrel{\raisebox{0.5\depth}{\scaleobj{0.5}{\langle \bullet \rangle}}}\Varid{m}~~.{}\<[E]%
\ColumnHook
\end{hscode}\resethooks
\end{proof}

\paragraph{Proving \eqref{eq:kc-mcomp}} \ensuremath{\Varid{f}\mathrel{\hstretch{0.7}{<\!\!\!=\!\!<}}(\Varid{g}\mathrel{\raisebox{0.5\depth}{\scaleobj{0.5}{\langle \bullet \rangle}}}\Varid{h})\mathrel{=}(\Varid{f}\mathbin{\cdot}\Varid{g})\mathrel{\hstretch{0.7}{<\!\!\!=\!\!<}}\Varid{h}}.
\begin{proof} We reason:
\begin{hscode}\SaveRestoreHook
\column{B}{@{}>{\hspre}l<{\hspost}@{}}%
\column{4}{@{}>{\hspre}l<{\hspost}@{}}%
\column{9}{@{}>{\hspre}l<{\hspost}@{}}%
\column{E}{@{}>{\hspre}l<{\hspost}@{}}%
\>[4]{}\Varid{f}\mathrel{\hstretch{0.7}{<\!\!\!=\!\!<}}(\Varid{g}\mathrel{\raisebox{0.5\depth}{\scaleobj{0.5}{\langle \bullet \rangle}}}\Varid{h}){}\<[E]%
\\
\>[B]{}\mathbin{=}{}\<[9]%
\>[9]{}\mbox{\commentbegin  definitions of \ensuremath{(\mathrel{\hstretch{0.7}{<\!\!\!=\!\!<}})}  \commentend}{}\<[E]%
\\
\>[B]{}\hsindent{4}{}\<[4]%
\>[4]{}(\Varid{f}\mathbin{\hstretch{0.7}{=\!\!<\!\!<}})\mathbin{\cdot}((\Varid{return}\mathbin{\cdot}\Varid{g})\mathbin{\hstretch{0.7}{=\!\!<\!\!<}})\mathbin{\cdot}\Varid{h}{}\<[E]%
\\
\>[B]{}\mathbin{=}{}\<[9]%
\>[9]{}\mbox{\commentbegin  by \eqref{eq:bind-comp-bind}  \commentend}{}\<[E]%
\\
\>[B]{}\hsindent{4}{}\<[4]%
\>[4]{}(((\Varid{f}\mathbin{\hstretch{0.7}{=\!\!<\!\!<}})\mathbin{\cdot}\Varid{return}\mathbin{\cdot}\Varid{g})\mathbin{\hstretch{0.7}{=\!\!<\!\!<}})\mathbin{\cdot}\Varid{h}{}\<[E]%
\\
\>[B]{}\mathbin{=}{}\<[9]%
\>[9]{}\mbox{\commentbegin  monad law \eqref{eq:monad-bind-ret}  \commentend}{}\<[E]%
\\
\>[B]{}\hsindent{4}{}\<[4]%
\>[4]{}((\Varid{f}\mathbin{\cdot}\Varid{g})\mathbin{\hstretch{0.7}{=\!\!<\!\!<}})\mathbin{\cdot}\Varid{h}{}\<[E]%
\\
\>[B]{}\mathbin{=}{}\<[9]%
\>[9]{}\mbox{\commentbegin  definition of \ensuremath{(\mathrel{\hstretch{0.7}{<\!\!\!=\!\!<}})}  \commentend}{}\<[E]%
\\
\>[B]{}\hsindent{4}{}\<[4]%
\>[4]{}(\Varid{f}\mathbin{\cdot}\Varid{g})\mathrel{\hstretch{0.7}{<\!\!\!=\!\!<}}\Varid{h}~~.{}\<[E]%
\ColumnHook
\end{hscode}\resethooks
\end{proof}

\paragraph{Proving \eqref{eq:mcomp-kc}} \ensuremath{\Varid{f}\mathrel{\raisebox{0.5\depth}{\scaleobj{0.5}{\langle \bullet \rangle}}}(\Varid{g}\mathrel{\hstretch{0.7}{<\!\!\!=\!\!<}}\Varid{h})\mathrel{=}(\Varid{f}\mathrel{\raisebox{0.5\depth}{\scaleobj{0.5}{\langle \bullet \rangle}}}\Varid{g})\mathrel{\hstretch{0.7}{<\!\!\!=\!\!<}}\Varid{h}}.
\begin{proof} We reason:
\begin{hscode}\SaveRestoreHook
\column{B}{@{}>{\hspre}l<{\hspost}@{}}%
\column{4}{@{}>{\hspre}l<{\hspost}@{}}%
\column{9}{@{}>{\hspre}l<{\hspost}@{}}%
\column{E}{@{}>{\hspre}l<{\hspost}@{}}%
\>[4]{}\Varid{f}\mathrel{\raisebox{0.5\depth}{\scaleobj{0.5}{\langle \bullet \rangle}}}(\Varid{g}\mathrel{\hstretch{0.7}{<\!\!\!=\!\!<}}\Varid{h}){}\<[E]%
\\
\>[B]{}\mathbin{=}{}\<[9]%
\>[9]{}\mbox{\commentbegin  definitions of \ensuremath{(\mathrel{\hstretch{0.7}{<\!\!\!=\!\!<}})} and \ensuremath{(\mathrel{\raisebox{0.5\depth}{\scaleobj{0.5}{\langle \bullet \rangle}}})}  \commentend}{}\<[E]%
\\
\>[B]{}\hsindent{4}{}\<[4]%
\>[4]{}((\Varid{return}\mathbin{\cdot}\Varid{f})\mathbin{\hstretch{0.7}{=\!\!<\!\!<}})\mathbin{\cdot}(\Varid{g}\mathbin{\hstretch{0.7}{=\!\!<\!\!<}})\mathbin{\cdot}\Varid{h}{}\<[E]%
\\
\>[B]{}\mathbin{=}{}\<[9]%
\>[9]{}\mbox{\commentbegin  by \eqref{eq:bind-comp-bind}  \commentend}{}\<[E]%
\\
\>[B]{}\hsindent{4}{}\<[4]%
\>[4]{}((((\Varid{return}\mathbin{\cdot}\Varid{f})\mathbin{\hstretch{0.7}{=\!\!<\!\!<}})\mathbin{\cdot}\Varid{g})\mathbin{\hstretch{0.7}{=\!\!<\!\!<}})\mathbin{\cdot}\Varid{h}{}\<[E]%
\\
\>[B]{}\mathbin{=}{}\<[9]%
\>[9]{}\mbox{\commentbegin  definition of \ensuremath{(\mathrel{\raisebox{0.5\depth}{\scaleobj{0.5}{\langle \bullet \rangle}}})}  \commentend}{}\<[E]%
\\
\>[B]{}\hsindent{4}{}\<[4]%
\>[4]{}((\Varid{f}\mathrel{\raisebox{0.5\depth}{\scaleobj{0.5}{\langle \bullet \rangle}}}\Varid{g})\mathbin{\hstretch{0.7}{=\!\!<\!\!<}})\mathbin{\cdot}\Varid{h}{}\<[E]%
\\
\>[B]{}\mathbin{=}{}\<[9]%
\>[9]{}\mbox{\commentbegin  definition of \ensuremath{(\mathrel{\hstretch{0.7}{<\!\!\!=\!\!<}})}  \commentend}{}\<[E]%
\\
\>[B]{}\hsindent{4}{}\<[4]%
\>[4]{}(\Varid{f}\mathrel{\raisebox{0.5\depth}{\scaleobj{0.5}{\langle \bullet \rangle}}}\Varid{g})\mathrel{\hstretch{0.7}{<\!\!\!=\!\!<}}\Varid{h}~~.{}\<[E]%
\ColumnHook
\end{hscode}\resethooks
\end{proof}

\paragraph{Proof of Lemma~\ref{lemma:hom-concat}}~
\begin{proof}
A Ping-pong proof.

{\sc Direction} $(\Rightarrow)$. Let \ensuremath{\Varid{h}\mathrel{=}\Varid{hom}\;(\oplus)\;(\Varid{h}\mathbin{\cdot}\Varid{wrap})\;\Varid{z}}, prove \ensuremath{\Varid{foldr}\;(\oplus)\;\Varid{z}\;(\Varid{map}\;\Varid{h}\;\Varid{xss})\mathrel{=}\Varid{h}\;(\Varid{concat}\;\Varid{xss})} by induction on \ensuremath{\Varid{xss}}.

{\sc Case} \ensuremath{\Varid{xss}\mathbin{:=}[\mskip1.5mu \mskip1.5mu]}:
\begin{hscode}\SaveRestoreHook
\column{B}{@{}>{\hspre}l<{\hspost}@{}}%
\column{7}{@{}>{\hspre}l<{\hspost}@{}}%
\column{E}{@{}>{\hspre}l<{\hspost}@{}}%
\>[7]{}\Varid{foldr}\;(\oplus)\;\Varid{z}\;(\Varid{map}\;\Varid{h}\;[\mskip1.5mu \mskip1.5mu]){}\<[E]%
\\
\>[B]{}\mathbin{=}{}\<[7]%
\>[7]{}\Varid{foldr}\;(\oplus)\;\Varid{z}\;[\mskip1.5mu \mskip1.5mu]{}\<[E]%
\\
\>[B]{}\mathbin{=}{}\<[7]%
\>[7]{}\Varid{z}{}\<[E]%
\\
\>[B]{}\mathbin{=}{}\<[7]%
\>[7]{}\Varid{h}\;(\Varid{concat}\;[\mskip1.5mu \mskip1.5mu])~~.{}\<[E]%
\ColumnHook
\end{hscode}\resethooks

{\sc Case} \ensuremath{\Varid{xss}\mathbin{:=}\Varid{xs}\mathbin{:}\Varid{xss}}:
\begin{hscode}\SaveRestoreHook
\column{B}{@{}>{\hspre}l<{\hspost}@{}}%
\column{7}{@{}>{\hspre}l<{\hspost}@{}}%
\column{9}{@{}>{\hspre}l<{\hspost}@{}}%
\column{E}{@{}>{\hspre}l<{\hspost}@{}}%
\>[7]{}\Varid{foldr}\;(\oplus)\;\Varid{z}\;(\Varid{map}\;\Varid{h}\;(\Varid{xs}\mathbin{:}\Varid{xss})){}\<[E]%
\\
\>[B]{}\mathbin{=}{}\<[7]%
\>[7]{}\Varid{h}\;\Varid{xs}\mathbin{\oplus}\Varid{foldr}\;(\oplus)\;\Varid{z}\;(\Varid{map}\;\Varid{h}\;\Varid{xss}){}\<[E]%
\\
\>[B]{}\mathbin{=}{}\<[9]%
\>[9]{}\mbox{\commentbegin  induction  \commentend}{}\<[E]%
\\
\>[B]{}\hsindent{7}{}\<[7]%
\>[7]{}\Varid{h}\;\Varid{xs}\mathbin{\oplus}\Varid{h}\;(\Varid{concat}\;\Varid{xss}){}\<[E]%
\\
\>[B]{}\mathbin{=}{}\<[9]%
\>[9]{}\mbox{\commentbegin  \ensuremath{\Varid{h}} homomorphism  \commentend}{}\<[E]%
\\
\>[B]{}\hsindent{7}{}\<[7]%
\>[7]{}\Varid{h}\;(\Varid{concat}\;(\Varid{xs}\mathbin{:}\Varid{xss}))~~.{}\<[E]%
\ColumnHook
\end{hscode}\resethooks

{\sc Direction} $(\Leftarrow)$. Assuming \ensuremath{\Varid{foldr}\;(\oplus)\;\Varid{z}\;(\Varid{map}\;\Varid{h}\;\Varid{xss})\mathrel{=}\Varid{h}\;(\Varid{concat}\;\Varid{xss})}, prove that \ensuremath{\Varid{h}\mathrel{=}\Varid{hom}\;(\oplus)\;(\Varid{h}\mathbin{\cdot}\Varid{wrap})\;\Varid{z}}.

{\sc Case} empty list:
\begin{hscode}\SaveRestoreHook
\column{B}{@{}>{\hspre}l<{\hspost}@{}}%
\column{7}{@{}>{\hspre}l<{\hspost}@{}}%
\column{9}{@{}>{\hspre}l<{\hspost}@{}}%
\column{E}{@{}>{\hspre}l<{\hspost}@{}}%
\>[7]{}\Varid{h}\;[\mskip1.5mu \mskip1.5mu]{}\<[E]%
\\
\>[B]{}\mathbin{=}{}\<[7]%
\>[7]{}\Varid{h}\;(\Varid{concat}\;[\mskip1.5mu \mskip1.5mu]){}\<[E]%
\\
\>[B]{}\mathbin{=}{}\<[9]%
\>[9]{}\mbox{\commentbegin  assumption  \commentend}{}\<[E]%
\\
\>[B]{}\hsindent{7}{}\<[7]%
\>[7]{}\Varid{foldr}\;(\oplus)\;\Varid{z}\;(\Varid{map}\;\Varid{h}\;[\mskip1.5mu \mskip1.5mu]){}\<[E]%
\\
\>[B]{}\mathbin{=}{}\<[7]%
\>[7]{}\Varid{z}~~.{}\<[E]%
\ColumnHook
\end{hscode}\resethooks

{\sc Case} singleton list: certainly \ensuremath{\Varid{h}\;[\mskip1.5mu \Varid{x}\mskip1.5mu]\mathrel{=}\Varid{h}\;[\mskip1.5mu \Varid{x}\mskip1.5mu]}.

{\sc Case} concatentation:
\begin{hscode}\SaveRestoreHook
\column{B}{@{}>{\hspre}l<{\hspost}@{}}%
\column{7}{@{}>{\hspre}l<{\hspost}@{}}%
\column{9}{@{}>{\hspre}l<{\hspost}@{}}%
\column{E}{@{}>{\hspre}l<{\hspost}@{}}%
\>[7]{}\Varid{h}\;(\Varid{xs}\mathbin{+\!\!\!\!\!+}\Varid{ys}){}\<[E]%
\\
\>[B]{}\mathbin{=}{}\<[7]%
\>[7]{}\Varid{h}\;(\Varid{concat}\;[\mskip1.5mu \Varid{xs},\Varid{ys}\mskip1.5mu]){}\<[E]%
\\
\>[B]{}\mathbin{=}{}\<[9]%
\>[9]{}\mbox{\commentbegin  assumption  \commentend}{}\<[E]%
\\
\>[B]{}\hsindent{7}{}\<[7]%
\>[7]{}\Varid{foldr}\;(\oplus)\;\Varid{z}\;(\Varid{map}\;\Varid{h}\;[\mskip1.5mu \Varid{xs},\Varid{ys}\mskip1.5mu]){}\<[E]%
\\
\>[B]{}\mathbin{=}{}\<[7]%
\>[7]{}\Varid{h}\;\Varid{xs}\mathbin{\oplus}(\Varid{h}\;\Varid{ys}\mathbin{\oplus}\Varid{z}){}\<[E]%
\\
\>[B]{}\mathbin{=}{}\<[7]%
\>[7]{}\Varid{h}\;\Varid{xs}\mathbin{\oplus}\Varid{h}\;\Varid{ys}~~.{}\<[E]%
\ColumnHook
\end{hscode}\resethooks
\end{proof}

\paragraph{Proof of Lemma~\ref{lemma:foldr-hom}}~
\begin{proof}
A Ping-pong proof.

{\sc Direction} $(\Leftarrow)$. We show that \ensuremath{\Varid{foldr}\;(\otimes)\;\Varid{z}\mathrel{=}\Varid{hom}\;(\oplus)\;(\mathbin{\otimes}\Varid{z})\;\Varid{z}}, provided that \ensuremath{\Varid{x}\mathbin{\otimes}(\Varid{y}\mathbin{\oplus}\Varid{w})\mathrel{=}(\Varid{x}\mathbin{\otimes}\Varid{y})\mathbin{\oplus}\Varid{w}}.

It is immediate that \ensuremath{\Varid{foldr}\;(\otimes)\;\Varid{z}\;[\mskip1.5mu \mskip1.5mu]\mathrel{=}\Varid{z}} around
\ensuremath{\Varid{foldr}\;(\otimes)\;\Varid{z}\;[\mskip1.5mu \Varid{x}\mskip1.5mu]\mathrel{=}\Varid{x}\mathbin{\otimes}\Varid{z}}.
For \ensuremath{\Varid{xs}\mathbin{+\!\!\!\!\!+}\Varid{ys}}, note that
\begin{align*}
  \ensuremath{\Varid{foldr}\;(\otimes)\;\Varid{z}\;(\Varid{xs}\mathbin{+\!\!\!\!\!+}\Varid{ys})} &= \ensuremath{\Varid{foldr}\;(\otimes)\;(\Varid{foldr}\;(\otimes)\;\Varid{z}\;\Varid{ys})\;\Varid{xs}} \mbox{~~.}
  \label{eq:fold-cat}
\end{align*}
The aim is thus to prove that
\begin{hscode}\SaveRestoreHook
\column{B}{@{}>{\hspre}l<{\hspost}@{}}%
\column{3}{@{}>{\hspre}l<{\hspost}@{}}%
\column{E}{@{}>{\hspre}l<{\hspost}@{}}%
\>[3]{}\Varid{foldr}\;(\otimes)\;(\Varid{foldr}\;(\otimes)\;\Varid{z}\;\Varid{ys})\;\Varid{xs}\mathbin{=}(\Varid{foldr}\;(\otimes)\;\Varid{z}\;\Varid{xs})\mathbin{\oplus}(\Varid{foldr}\;(\otimes)\;\Varid{z}\;\Varid{ys})~~.{}\<[E]%
\ColumnHook
\end{hscode}\resethooks
We perform an induction on \ensuremath{\Varid{xs}}. The case when \ensuremath{\Varid{xs}\mathbin{:=}[\mskip1.5mu \mskip1.5mu]} trivially
holds. For \ensuremath{\Varid{xs}\mathbin{:=}\Varid{x}\mathbin{:}\Varid{xs}}, we reason:
\begin{hscode}\SaveRestoreHook
\column{B}{@{}>{\hspre}l<{\hspost}@{}}%
\column{7}{@{}>{\hspre}l<{\hspost}@{}}%
\column{8}{@{}>{\hspre}l<{\hspost}@{}}%
\column{E}{@{}>{\hspre}l<{\hspost}@{}}%
\>[7]{}\Varid{foldr}\;(\otimes)\;(\Varid{foldr}\;(\otimes)\;\Varid{z}\;\Varid{ys})\;(\Varid{x}\mathbin{:}\Varid{xs}){}\<[E]%
\\
\>[B]{}\mathbin{=}{}\<[7]%
\>[7]{}\Varid{x}\mathbin{\otimes}\Varid{foldr}\;(\otimes)\;(\Varid{foldr}\;(\otimes)\;\Varid{z}\;\Varid{ys})\;\Varid{xs}{}\<[E]%
\\
\>[B]{}\mathbin{=}{}\<[8]%
\>[8]{}\mbox{\commentbegin  induction  \commentend}{}\<[E]%
\\
\>[B]{}\hsindent{7}{}\<[7]%
\>[7]{}\Varid{x}\mathbin{\otimes}((\Varid{foldr}\;(\otimes)\;\Varid{z}\;\Varid{xs})\mathbin{\oplus}(\Varid{foldr}\;(\otimes)\;\Varid{z}\;\Varid{ys})){}\<[E]%
\\
\>[B]{}\mathbin{=}{}\<[8]%
\>[8]{}\mbox{\commentbegin  assumption: \ensuremath{\Varid{x}\mathbin{\otimes}(\Varid{y}\mathbin{\oplus}\Varid{w})\mathrel{=}(\Varid{x}\mathbin{\otimes}\Varid{y})\mathbin{\oplus}\Varid{w}}  \commentend}{}\<[E]%
\\
\>[B]{}\hsindent{7}{}\<[7]%
\>[7]{}(\Varid{x}\mathbin{\otimes}(\Varid{foldr}\;(\otimes)\;\Varid{z}\;\Varid{xs}))\mathbin{\oplus}(\Varid{foldr}\;(\otimes)\;\Varid{z}\;\Varid{ys}){}\<[E]%
\\
\>[B]{}\mathbin{=}{}\<[7]%
\>[7]{}(\Varid{foldr}\;(\otimes)\;\Varid{z}\;(\Varid{x}\mathbin{:}\Varid{xs}))\mathbin{\oplus}(\Varid{foldr}\;(\otimes)\;\Varid{z}\;\Varid{ys})~~.{}\<[E]%
\ColumnHook
\end{hscode}\resethooks

{\sc Direction} $(\Rightarrow)$. Given
\ensuremath{\Varid{foldr}\;(\otimes)\;\Varid{z}\mathrel{=}\Varid{hom}\;(\oplus)\;(\mathbin{\otimes}\Varid{z})\;\Varid{z}},
prove that \ensuremath{\Varid{x}\mathbin{\otimes}(\Varid{y}\mathbin{\oplus}\Varid{w})\mathrel{=}(\Varid{x}\mathbin{\otimes}\Varid{y})\mathbin{\oplus}\Varid{w}} for \ensuremath{\Varid{y}} and \ensuremath{\Varid{w}} in the range of \ensuremath{\Varid{foldr}\;(\otimes)\;\Varid{z}}.

Let \ensuremath{\Varid{y}\mathrel{=}\Varid{foldr}\;(\otimes)\;\Varid{z}\;\Varid{xs}} and \ensuremath{\Varid{w}\mathrel{=}\Varid{foldr}\;(\otimes)\;\Varid{z}\;\Varid{ys}} for some
\ensuremath{\Varid{xs}} and \ensuremath{\Varid{ys}}. We reason:
\begin{hscode}\SaveRestoreHook
\column{B}{@{}>{\hspre}l<{\hspost}@{}}%
\column{7}{@{}>{\hspre}l<{\hspost}@{}}%
\column{9}{@{}>{\hspre}l<{\hspost}@{}}%
\column{E}{@{}>{\hspre}l<{\hspost}@{}}%
\>[7]{}\Varid{x}\mathbin{\otimes}(\Varid{y}\mathbin{\oplus}\Varid{w}){}\<[E]%
\\
\>[B]{}\mathbin{=}{}\<[7]%
\>[7]{}\Varid{x}\mathbin{\otimes}(\Varid{foldr}\;(\otimes)\;\Varid{z}\;\Varid{xs}\mathbin{\oplus}\Varid{foldr}\;(\otimes)\;\Varid{z}\;\Varid{ys}){}\<[E]%
\\
\>[B]{}\mathbin{=}{}\<[9]%
\>[9]{}\mbox{\commentbegin  since \ensuremath{\Varid{foldr}\;(\otimes)\;\Varid{z}\mathrel{=}\Varid{hom}\;(\oplus)\;(\mathbin{\otimes}\Varid{z})\;\Varid{z}}  \commentend}{}\<[E]%
\\
\>[B]{}\hsindent{7}{}\<[7]%
\>[7]{}\Varid{x}\mathbin{\otimes}(\Varid{foldr}\;(\otimes)\;\Varid{z}\;(\Varid{xs}\mathbin{+\!\!\!\!\!+}\Varid{ys})){}\<[E]%
\\
\>[B]{}\mathbin{=}{}\<[7]%
\>[7]{}\Varid{foldr}\;(\otimes)\;\Varid{z}\;(\Varid{x}\mathbin{:}\Varid{xs}\mathbin{+\!\!\!\!\!+}\Varid{ys}){}\<[E]%
\\
\>[B]{}\mathbin{=}{}\<[9]%
\>[9]{}\mbox{\commentbegin  since \ensuremath{\Varid{foldr}\;(\otimes)\;\Varid{z}\mathrel{=}\Varid{hom}\;(\oplus)\;(\mathbin{\otimes}\Varid{z})\;\Varid{z}}  \commentend}{}\<[E]%
\\
\>[B]{}\hsindent{7}{}\<[7]%
\>[7]{}\Varid{foldr}\;(\otimes)\;\Varid{z}\;(\Varid{x}\mathbin{:}\Varid{xs})\mathbin{\oplus}\Varid{foldr}\;(\otimes)\;\Varid{z}\;\Varid{ys}{}\<[E]%
\\
\>[B]{}\mathbin{=}{}\<[7]%
\>[7]{}(\Varid{x}\mathbin{\otimes}\Varid{foldr}\;(\otimes)\;\Varid{z}\;\Varid{xs})\mathbin{\oplus}\Varid{foldr}\;(\otimes)\;\Varid{z}\;\Varid{ys}{}\<[E]%
\\
\>[B]{}\mathbin{=}{}\<[7]%
\>[7]{}(\Varid{x}\mathbin{\otimes}\Varid{y})\mathbin{\oplus}\Varid{w}~~.{}\<[E]%
\ColumnHook
\end{hscode}\resethooks
\end{proof}
%

\end{document}